\documentclass[aip,cha,
 amsmath,amssymb,reprint,
]{revtex4-1}

\usepackage{graphicx}
\usepackage{dcolumn}
\usepackage{bm}

\usepackage[utf8]{inputenc}
\usepackage[T1]{fontenc}
\usepackage{mathptmx}
\usepackage{etoolbox}

\usepackage{subcaption,booktabs,amsthm,amsfonts,xcolor}

\newcommand{\ovl}[1]{{\overline{#1}}}

\newcommand{\op}{\operatorname}

\renewcommand{\Re}{\operatorname{Re}}

\newcommand{\norm}[1]{\left\|{#1}\right\|}
\newcommand{\normno}[1]{\|{#1}\|}

\newcommand{\CC}{\mathbb C}

\newcommand{\RR}{\mathbb R}

\newcommand{\cB}{\mathcal{B}}

\newcommand{\cm}{\mathrm{cm}}
\newcommand{\s}{\mathrm{s}}

\newcommand{\dave}[1]{#1}

\newcommand{\DtN}{\operatorname{DtN}}

\newtheorem{lemma}{Lemma}

\newtheorem{claim}{Claim}
\newtheorem{theorem}{Theorem}
\newtheorem{proposition}{Proposition}

\theoremstyle{definition}
\newtheorem{remark}{Remark}
\newtheorem{definition}{Definition}
\newtheorem{example}{Example}

\makeatletter
\def\@email#1#2{%
 \endgroup
 \patchcmd{\titleblock@produce}
  {\frontmatter@RRAPformat}
  {\frontmatter@RRAPformat{\produce@RRAP{*#1\href{mailto:#2}{#2}}}\frontmatter@RRAPformat}
  {}{}
}%
\makeatother
\begin{document}


\title[Extending the Droplet-Wave Statistical Correspondence in Walking Droplet Dynamics]{Extending the Droplet-Wave Statistical Correspondence in Walking Droplet Dynamics\vspace{0.5em}}
\author{S. Mao}
\author{D. Darrow*}%
 \email{ddarrow@mit.edu.}
\affiliation{Department of Mathematics, Massachusetts Institute of Technology, Cambridge, MA, USA
}%

\date{\today}

\begin{abstract}
Walking droplets---millimetric oil droplets that self-propel across the surface of a vibrating fluid bath---exhibit striking emergent statistics that remain only partially understood. In particular, in a variety of experiments, a robust correspondence has been observed between the droplet's statistical distribution and the time-average of the wave field that guides it. M.~Durey, P.~A.~Milewski, and J.~W.~M.~Bush, \emph{Chaos} \textbf{28}, 096108 (2018) \dave{rigorously established such a correspondence for single-droplet systems with a single, instantaneous droplet-bath impact during each vibration period, but numerical and experimental evidence suggests that the correspondence should hold far more broadly. Laboratory droplet systems, for instance, often exhibit complex bouncing modes that do not adhere to these hypotheses.} We attempt to complete this program in the present work, rigorously extending this statistical correspondence to account for \dave{arbitrary droplet-bath impact models, multi-droplet interactions, and non-resonant bouncing}. We investigate this correspondence numerically in systems of one and two droplets in 1-D geometries, and \dave{we highlight how} the time-averaged wave field can distinguish between correlated and uncorrelated pairs of droplets.
\end{abstract}

\maketitle

\begin{quotation}
    `Walking' oil droplets provide a compelling example of coupled particle-wave dynamics at the macro-scale, giving rise to rich, wave-like statistical patterns reminiscent of quantum systems. At the heart of many of these `hydrodynamic quantum analogues' is a widely-observed correspondence between the statistics of a walking droplet and the time-average of the wave it bounces on, \dave{which has been described as analogous to the \emph{Born rule} of quantum mechanics}. This correspondence has been proven rigorously under strict assumptions on the droplet-wave interaction, but experimental and numerical evidence suggests that it should hold more broadly. We attempt to complete this program in the present work, rigorously establishing such a correspondence for \dave{nearly arbitrary droplet models, multi-droplet interactions, and droplets bouncing out-of-phase with their guiding wave}.
\end{quotation}

\section{Introduction}
A millimetric oil droplet can be made to bounce on the surface of a vibrating fluid bath; with sufficiently fast vertical vibration, the droplet does not have enough time to deplete the thin air layer below it and coalesce into the bath. 
Under the right conditions, the droplet's self-generated waves can propel it forward, allowing it to travel across the surface of the bath \citep{Couder2005}. Such \emph{walking droplets} offer a striking example of coupled particle-wave dynamics at the macro-scale. Among other results, walking droplets have been shown to exhibit diffraction in analogue single- and double-slit experiments~\citep{Couder2006,Ellegaard2020,Pucci2018,Pucci2024}, to `tunnel' past submerged barriers~\citep{PhysRevLett.102.240401,Nachbin2017}, to attain quantized orbits~\citep{fort_2010,harris_bush_2014}, and to form regular \dave{lattice patterns}~\citep{Eddi_2009}.

In many of these experiments, walking droplets exhibit coherent statistical patterns despite their complex dynamics. In particular, in a wide variety of droplet-based `quantum analogues'~\citep{hqa}, the probability density function (PDF) of the droplet's position is observed to align closely with the time-average of its guiding wave at the start (or another fixed phase\footnote{Such a stroboscopic average is necessary, as the wave field vanishes when averaged over a continuous period.}) of each period of oscillation \citep{PhysRevE.88.011001,Saenz2018,durey_chaos,Durey_Milewski_Wang_2020,kutz_pilot-wave_2023,Abraham_2024,Evans_2025}:
\begin{equation}\label{eq:MWF}
    \ovl{\eta}(x) = \lim_{N\to\infty} \frac{1}{N}\sum_{n=0}^N \eta(x,nT_F),
\end{equation}
where $T_F$ is the Faraday period of the wave and $\eta(x,t)$ is the wave height field at time $t$. The object $\ovl{\eta}(x)$ is often called the \emph{mean wave field} (MWF) of the given system. Figure~\ref{fig:elliptic} shows the visual similarity between the PDF and the mean wave \emph{slope}---related to the MWF but easier to detect from an overhead camera~\citep{Saenz2018}---in an elliptical domain.

\begin{figure*}
     \centering
     \begin{subfigure}{0.33\textwidth}
         \centering
         \includegraphics[scale=0.5]{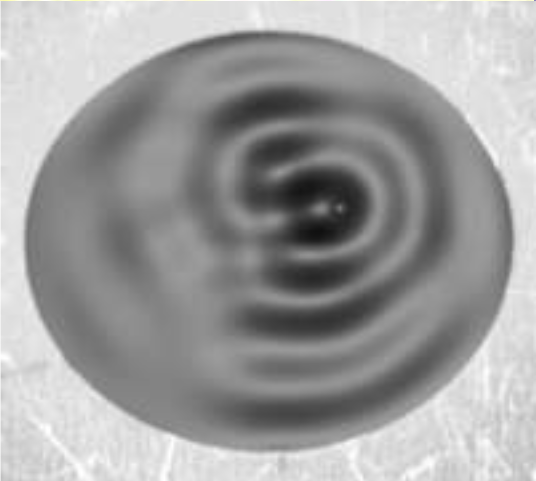}
         \caption{}
     \end{subfigure}%
     \begin{subfigure}{0.33\textwidth}
         \centering
         \includegraphics[scale=0.5]{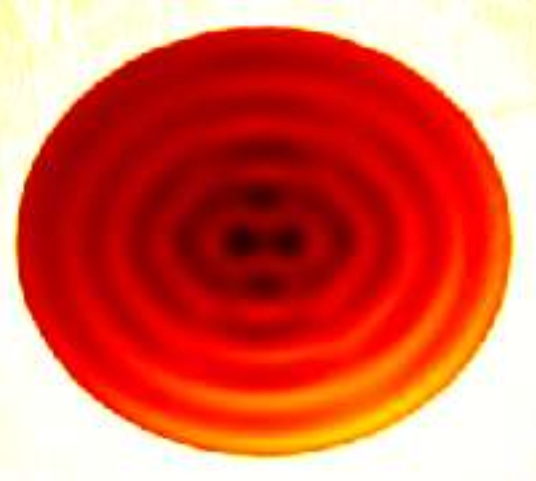}
         \caption{}
     \end{subfigure}%
     \begin{subfigure}{0.33\textwidth}
         \centering
         \includegraphics[scale=0.5]{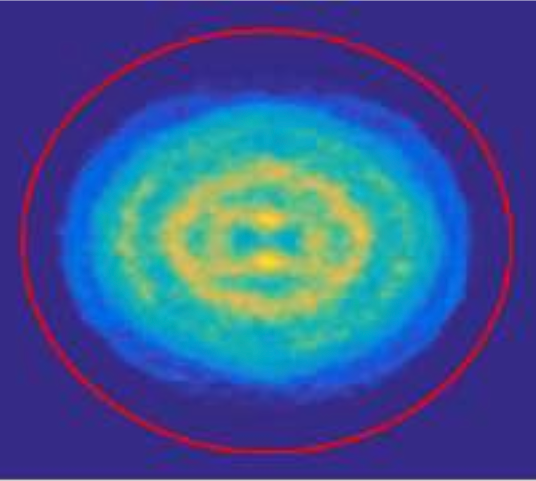}
         \caption{}
     \end{subfigure}
     \caption{\textbf{(a)} Snapshot of the instantaneous wave field of a walking droplet in an elliptical bath. Shading corresponds to the slope of the wave field, rather than its height, as the former is more easily detected from an overhead camera~\citep{Saenz2018}. \textbf{(b)} The time-averaged slope of the wave field, which is closely related to the time-averaged wave height (i.e., the \emph{mean wave field}). \textbf{(c)} Probability density function of the droplet, showing a strong visual similarity with the time-averaged wave slope. A correspondence between the probability density and mean wave field has been widely observed, but only proven under strict assumptions on the droplet-wave interaction. Images modified with permission from P. J. S\'aenz, T. Cristea-Platon, and J. W. M. Bush, ``Statistical projection
effects in a hydrodynamic pilot-wave system,'' Nature Physics 14, 315–319
(2018). Copyright 2018 Springer Nature.}
     \label{fig:elliptic}
\end{figure*}

\dave{An important step towards understanding this PDF-MWF correspondence was made by \citet{durey_chaos}. The authors were able to rigorously prove that a PDF-MWF correspondence must hold in single-droplet systems with a single, instantaneous droplet-bath impact at the start of each vibration period}. They find that if the horizontal droplet dynamics are ergodic (or periodic) with a stationary PDF $\rho(x)$ for the particle position, then the MWF $\bar{\eta}(x)$ satisfies
\begin{equation}
\bar{\eta}(x)=\int\eta_B(x,y)\rho(y)~dy,\label{eq:durey}
\end{equation}
where $\eta_B(x,y)$ is the MWF corresponding to a stationary bouncing droplet centered at $y$. \dave{The authors studied in depth the case of a uniform, unbounded fluid bath, and \citet{Durey_Milewski_Wang_2020} have applied their results similarly to investigate a droplet in a circular corral. More recently, \citet{Abraham_2024} have applied the same PDF-MWF correspondence to study the role that the mean wave field plays in the `Anderson localization' of walking droplets over randomized topography. In general, the relation~\eqref{eq:durey} has been described as analogous to the \emph{Born rule} of quantum mechanics~\cite{griffiths2017introduction,kutz_pilot-wave_2023}, which relates the spatial distribution $\rho(x)$ of a quantum particle's measured position to the profile of its wavefunction $\psi$:}
\begin{equation}\label{eq:born}
	\rho(x) = |\psi(x)|^2.
\end{equation}

\dave{Evidence suggests that a PDF-MWF correspondence holds more broadly than the hypotheses of Durey, Milewski, and Bush allow}. For instance, \citet{kutz_pilot-wave_2023} report a PDF-MWF correspondence in numerical simulations of inhomogeneous one-dimensional droplet systems with a finite droplet-bath impact time, and \citet{Evans_2025} have observed similar behavior when a droplet bounces \emph{aperiodically}, out of phase from its guiding wave. \dave{More basically, droplets in laboratory settings exhibit a wide variety of bouncing modes that violate these hypotheses, including modes with extended droplet-bath impact times and with multiple impacts per forcing period~\citep{molacek_bush_2013}. To better understand walking droplet statistics in practical settings---and to strengthen the analogy with Born's rule---we require a more general version of the correspondence~\eqref{eq:durey}.}

In the present work, we leverage the theory of \emph{evolution systems}~\citep{Pazy1963} to establish a PDF-MWF correspondence in a highly general setting, completing the program laid out by \citet{durey_chaos} and corroborating these experimental and numerical observations. We find the following result, stated rigorously in Appendix~\ref{sec:main} as Theorem~\ref{thm:main_infty}:
\begin{claim}\label{claim:main_heuristic}
    Fix a domain $\Omega\subset\RR^2$, and write $\eta:\Omega\to\RR$ for the time-evolving wave field. Suppose the system is being vibrated below the Faraday threshold, so the wave field decays exponentially in the absence of droplets. Suppose there are one or more walking droplets in $\Omega$, and that the wave-droplet interaction in each period of vibration\footnote{\dave{We refer here to the period at which the fluid bath is \emph{driven}, or, if desired, a multiple thereof; note that the wave field and vertical droplet dynamics can exhibit periodic dynamics at multiples of this driving period~\citep{molacek_bush_2013}.}} is parametrized \dave{(or well-approximated)} by a finite collection of parameters $\vec{\alpha}=\{\alpha_1,...,\alpha_N\}$; for instance, these parameters can quantify positions, velocities, or phase offsets of droplet bounces.
    
    If the \dave{droplet's trajectory exhibits a stationary probability distribution $\rho(\vec{\alpha})$, then the mean wave field $\ovl{\eta}(x)$ is given by}
    \begin{equation}\label{eq:main_heuristic}
        \ovl{\eta}(x) = \int_{\RR^N}\eta_B(x,\vec{\alpha})\rho(\vec{\alpha})\,d\vec{\alpha},
    \end{equation}
     where $\eta_B(x,\vec{\alpha}_0)$ is the mean wave field at $x\in\Omega$ corresponding to a constant bouncing configuration $\vec{\alpha}(t)\equiv\vec{\alpha}_0$. 
\end{claim}
\dave{We note that one does not necessarily need to resolve the full droplet dynamics for all time in order to apply this result. For instance, a common approximation is to assume that the droplet's vertical dynamics take a prescribed form, and that each droplet impact is well-localized at a position $x_p\in\Omega$. Under this approximation, Claim~\ref{claim:main_heuristic} would involve only the spatial probability distribution $\rho(x_p)$ of the droplet, much like the theorem of \citet{durey_chaos}. We discuss other cases of interest in Section~\ref{sec:apps}, including multi-droplet systems and non-resonant droplet bouncing.}


\dave{It is also worth mentioning that we have dropped the hypothesis of \emph{ergodicity} used by \citet{durey_chaos}. As discussed in their work, it has been widely observed that confined droplet systems often give rise to stationary probability distributions for the droplet's position~\citep{perrard_2014,PhysRevE.88.011001,Saenz2018}, but do not always exhibit global ergodicity. In short, the stationary probability distribution exhibited by a particular droplet's trajectory may depend on its initial condition. Fortunately, the \emph{ergodic decomposition theorem} states that, under very general conditions, any stationary distribution of a dynamical system can be decomposed as a sum over ergodic subsets of the global phase space~\citep{viana2016foundations}; consequently, any particular, generic droplet trajectory will tend toward one of these ergodic subsets, and the ergodic theorem can be applied as before. Notably, this argument accounts for the two particular cases studied by Durey, Milewski, and Bush: global ergodicity and periodic cycles}.

\dave{In Section~\ref{sec:model}, we review the mathematical models used for walking droplets, and we discuss how our results apply to the full range of droplet models. In Section~\ref{sec:numerics}, we support Claim~\ref{claim:main_heuristic} numerically using the one-dimensional droplet model of \citet{PhysRevFluids.7.093604}, which involves a finite-time droplet-bath impact}. We verify that~\eqref{eq:durey} holds for single droplet systems in different geometries and with different vibration amplitudes, and we demonstrate in a two-droplet system how the mean wave field can be used to study the correlation between multiple droplets. \dave{In Section~\ref{sec:apps}, we discuss key applications of Claim~\ref{claim:main_heuristic} to current themes in walking droplet research, including multi-droplet statistics, refined models of the droplet-bath interaction, and droplets bouncing out-of-phase with their guiding wave}.  We prove Claim~\ref{claim:main_heuristic} rigorously in Appendix~\ref{sec:main}, first illustrating our argument in the setting of finite-dimensional ODEs and then extending to the full PDE case. There, we also highlight in Example~\ref{ex:molacek} how the claim applies rigorously to the droplet model of \citet{milewski_galeano-rios_nachbin_bush_2015}.

\begin{figure*}
    \centering
    \begin{subfigure}{0.33\textwidth}
        \centering
        \includegraphics[width=\linewidth]{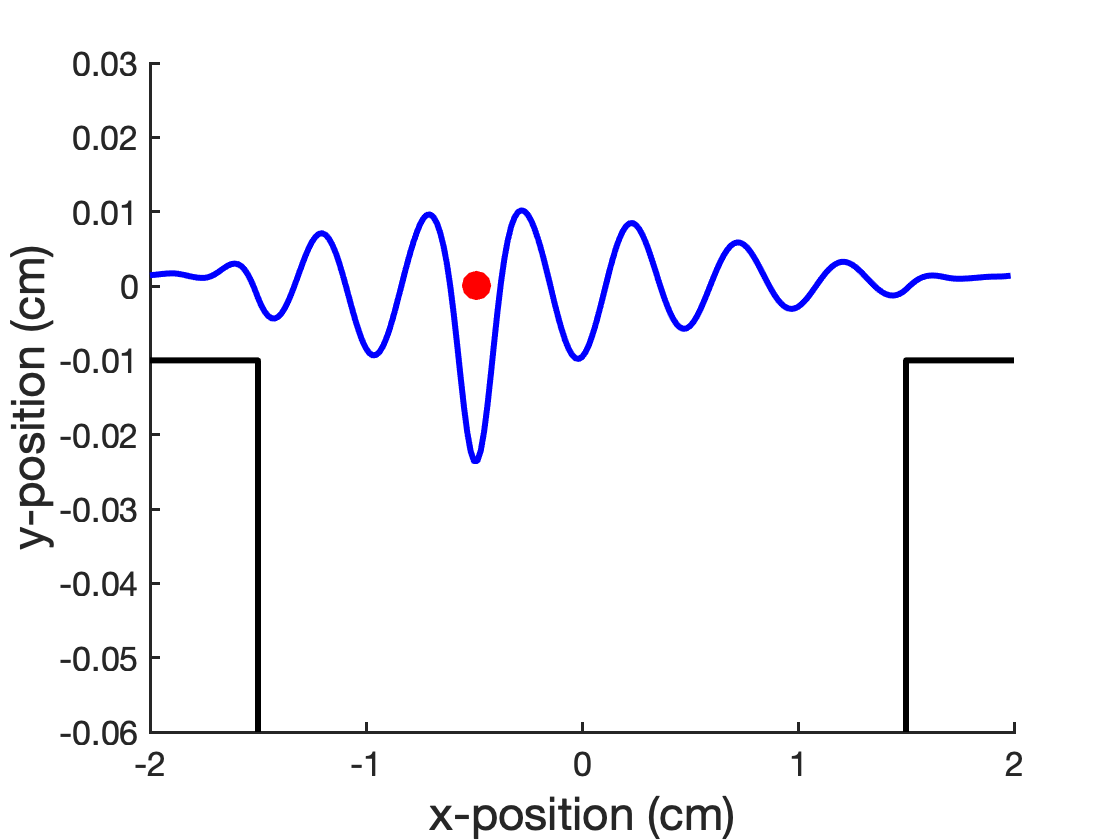}
        \caption{}
        \label{fig:snapshot}
    \end{subfigure}
    \begin{subfigure}{0.33\textwidth}
        \centering\includegraphics[width=\linewidth]{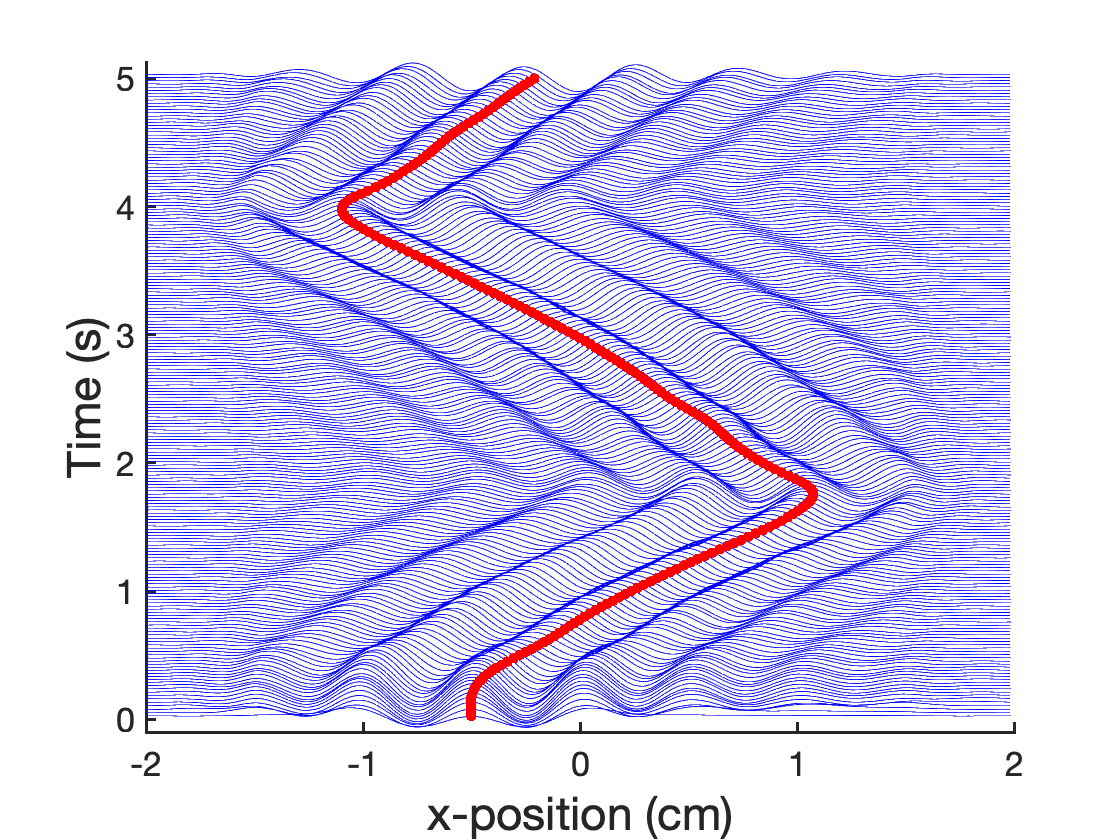}
        \caption{}\label{fig:trajectory}
    \end{subfigure}%
    \begin{subfigure}{0.33\textwidth}
        \centering\includegraphics[width=\linewidth]{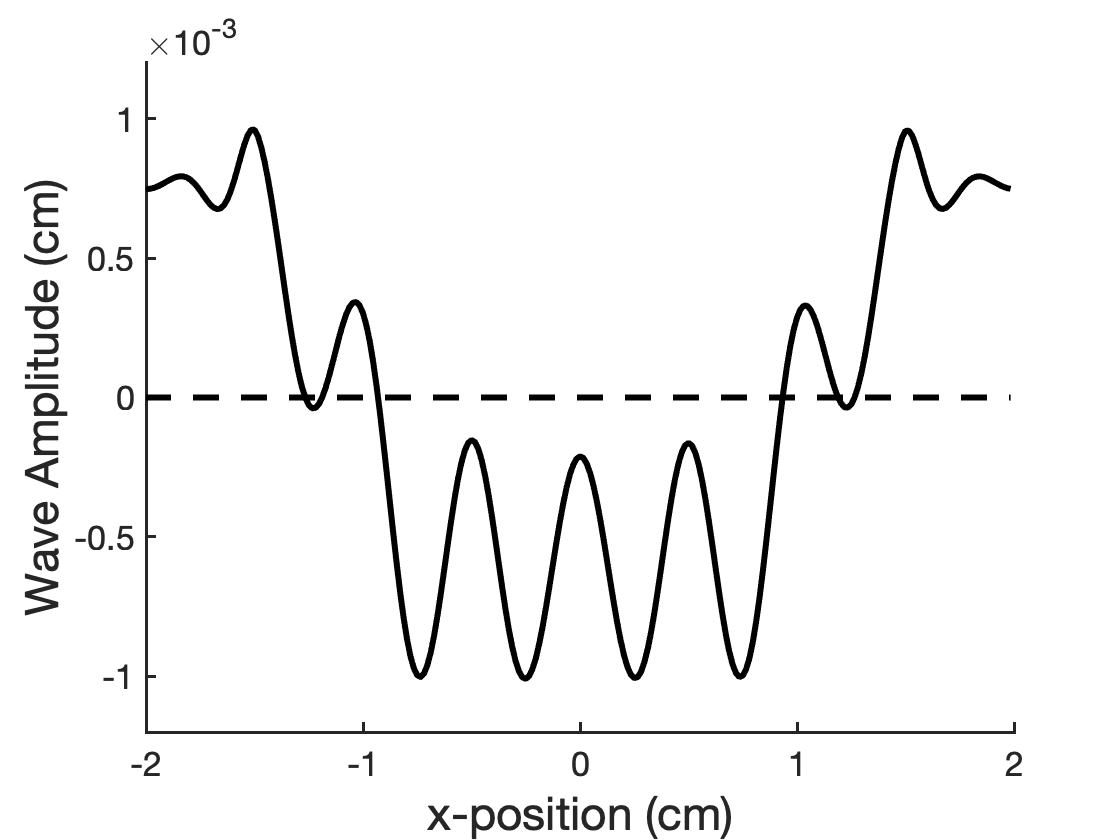}
        \caption{}\label{fig:exdropC}
    \end{subfigure}
    \caption{\textbf{(a)} A snapshot of the droplet model, showing the bath topography (black, partially shown), the wave field (blue), and the droplet (red). \textbf{(b)} A droplet (red) starts at $-0.5~\mathrm{cm}$ and traverses the $3~\mathrm{cm}$ domain with memory parameter $\Gamma/\Gamma_F=0.85$. Snapshots of the underlying wavefield (blue) are shown at each Faraday period. \textbf{(c)} The long-term mean wave field (black) is calculated by averaging the value of the wave field at the start of each Faraday period.}
    \label{fig:exdrop}
\end{figure*}

\section{Background on Walking Droplet Modeling}\label{sec:model}

\dave{In the present section, we review the various analytical models that have been developed for the walking droplet system, and discuss how Claim~\ref{claim:main_heuristic} applies to each of them.}

Although early models attempted to reduce the walking droplet system to a single ODE for the droplet's horizontal motion \citep{Couder2005, protiere2006}, more refined models account for the \emph{path memory} of the droplet~\citep{Eddi2011}, i.e., how the droplet's current motion depends on the waves generated by its previous bounces. One method of approach, carried out by \citet{Eddi2011}, \citet{MolacekBush2012}, and \citet{Oza_Rosales_Bush_2013}, is to fold the contribution of the guiding wave directly into the trajectory equation, recovering a nonlinear \emph{integro-differential} equation for the system. Recent models have taken such an approach to account for non-resonant droplet bouncing~\citep{Primkulov_2025}, as discussed in Section~\ref{sec:apps}. 

\dave{Although not manifest in their evolution equations, these models typically make predictions for the wave field consistent with the results presented here. As a simple example, consider the \emph{stroboscopic} model of \citet{Oza_Rosales_Bush_2013}, where the wave height field is given by
\[h(x,t) = \sum_{n=-\infty}^{\lfloor t/T_F \rfloor}
      AJ_0\big(k_F |x - x_p(nT_F)|\big)
      e^{-(t - nT_F)/\tau}
    .\]
Here, $T_F$ and $k_F$ are the Faraday period and wavenumber, respectively, $x_p$ is the horizontal position of the particle, and $\tau>0$ is a decay timescale for the wave field. It is not hard to verify that $h$ satisfies the PDE
\begin{gather*}
    \partial_t h = -\tau^{-1}h + P_0,\\
    P_0(x,t)=\sum_n A\delta(t-nT_F)\,J_0(k_F|x-x_p(nT_F)|),
\end{gather*}
writing $\delta(t-t_0)$ for the Dirac delta function at a given time point $t_0$. This PDE satisfies the rigorous hypotheses\footnote{\dave{Indeed, the instantaneous droplet-bath impact utilized here does \emph{not} satisfy our hypotheses as stated, but it is straightforward to adapt the theorem to allow for instantaneous impacts, by combining our argument with that of \citet{durey_chaos}.}} laid out in Appendix~\ref{sec:main}, so we deduce the PDF-MWF correspondence stated in Claim~\ref{claim:main_heuristic}.}

A more flexible approach \dave{to walking droplet modeling}, introduced by~\citet{milewski_galeano-rios_nachbin_bush_2015}, is to explicitly model the motion of the bath surface. They proposed the following set of quasi-potential-flow equations~(adapted from \citet{dias2008}) for the free-surface displacement $\eta$ and velocity potential $\phi$ of the oil bath:
\begin{equation}\label{eq:dureyeqs}
    \begin{gathered}
        \partial_t\phi = -g(t)\eta+2\nu\nabla^2\phi+\frac{\sigma}{\rho}\nabla^2\eta-\frac{P_0}{\rho},\\
        \partial_t\eta=\DtN(\phi)+2\nu\nabla^2\eta,
    \end{gathered}
\end{equation}
with appropriate boundary conditions. Here, $\nu$, $\sigma$, and $\rho$ represent fluid viscosity, surface tension, and density, respectively, and $P_0=P_0(x,t)$ is the pressure imposed by the droplet on the bath, nonzero only when the droplet and bath are in contact. The coefficient $g(t) = -g(1-\Gamma\cos(4\pi t/T_F))$ represents vertical shaking of the bath, with gravitational acceleration $g$ and non-dimensionalized vibration amplitude $\Gamma$. \dave{The operator $\operatorname{DtN}$ is the geometry-dependent Dirichlet-to-Neumann map for the bath~\citep{Quarteroni1991}, which maps the velocity potential $\phi$ to the vertical fluid velocity $w=\phi_z$ at the surface of the bath}. 

\dave{As a note, the vibration amplitude $\Gamma$ is typically fixed slightly below the Faraday threshold $\Gamma_F$, above which the bath would destabilize into a field of Faraday waves~\citep{benjaminursell}. Below $\Gamma_F$, the droplet's bouncing can still excite waves at the Faraday scale; the `memory parameter' $\Gamma/\Gamma_F$ determines the decay rate of such waves~\citep{Eddi2011}.}

\dave{In order to model the full system, the wave field equation \eqref{eq:dureyeqs} must be coupled with an appropriate trajectory equation for the droplet. Different trajectory equations are used for different purposes; that of \citet{molacek_bush_2013} accurately models the droplet's vertical bouncing, for instance, where that of \citet{Nachbin2017} abstracts away the vertical dynamics and assumes a periodic droplet-bath impact of fixed duration. In Example~\ref{ex:molacek} of Appendix~\ref{sec:main}, we prove that Claim~\ref{claim:main_heuristic} applies to models invoking the equation \eqref{eq:dureyeqs} \emph{regardless} of the trajectory equation. }

Alternatives to~\eqref{eq:dureyeqs} have been proposed by various authors. Notably, \citet{Faria_2016} \dave{posited an ad-hoc, approximate} form of the Dirichlet-to-Neumann operator \dave{for piecewise-constant bath topographies}, which has seen wide success in numerical simulations~\citep{Pucci2016diffract,Saenz2018,Harris2018}. In a different direction, \citet{Nachbin2017} were able to treat the operator exactly in one-dimensional baths by studying the system under a conformal map. \dave{Claim~\ref{claim:main_heuristic} applies as before to all of these alternatives; in fact,} we employ the model of \citet{Nachbin2017} in our numerical investigations in Section~\ref{sec:numerics}.

\begin{table}
    \centering\footnotesize
    \begin{tabular}{ccc}
        Variable & Value & Description \\ 
        \hline\hline\vspace{-1em}\\
        $g$ & $981~\cm~\s^{-2}$ & Gravity\\
        $\omega$ & $80\times 2\pi~\mathrm{rad}~\s^{-1}$ & Vibration frequency\\
        $T_F$ & $0.025~\s$ & Faraday period \\
        $\sigma$ & $20.9~\mathrm{dyn}~\cm^{-1}$ & Surface tension\\
        $\rho$ & $0.95~\mathrm{g}~\cm^{-3}$ & Fluid density\\
        $\nu$ & $0.16~\cm^2~\s^{-1}$ & Fluid viscosity\\
        $c$ & $0.25$ & Drag-to-force ratio~\eqref{eq:drag}\\
        $R_0$ & $0.035~\cm$ & Droplet radius\\
    \end{tabular}
    \vspace{0.1in}
    \caption{Physical parameters used in our simulations, following~\citet{PhysRevFluids.7.093604}.}
    \label{tab:constants}
\end{table}

\section{Numerical PDF-MWF Correspondence}\label{sec:numerics}

We here explore the consequences of Claim~\ref{claim:main_heuristic} numerically. Namely, we investigate the PDF-MWF correspondence within the one-dimensional model of \citet{Nachbin2017}, which exactly resolves the topography of the fluid bath using an appropriate conformal map. \dave{Following Nachbin, Milewski, and Bush, we make use of the wave equation \eqref{eq:dureyeqs} in 1-D, along with the following horizontal trajectory equation\footnote{\dave{We note that the present simulations (as well as those of Nachbin, Milewski, and Bush) evolve the droplet according to the gradient of the true, instantaneous wave field. An alternate approach was taken by~\citet{molacek_bush_2013} and \citet{milewski_galeano-rios_nachbin_bush_2015}, employing a `hypothetical' wave field not affected by the current droplet impact.}}:
\begin{equation}\label{eq:trajectory}
	m\ddot{x}_p + D(t)\dot{x}_p = -F(t)\nabla\eta(x_p,t),
\end{equation}
where $x_p$ is the horizontal droplet position, $m$ is its mass, $D(t)$ is a time-dependent drag term, and 
\[F(t) = \frac{8\pi^2 mg}{\omega T_F}\mathbf{1}_{0<t<T_F/4}(t)\sin(4\pi t/T_F)\] 
is a time-dependent forcing, nonzero only when the droplet and bath are in contact~\citep{molacek_bush_2013}. Neglecting air resistance, we fix
\begin{equation}\label{eq:drag}
	D(t) = c\sqrt{\frac{\rho R_0}{\sigma}}F(t),
\end{equation}
where $c>0$ is a scaling parameter and $R_0$ is the droplet radius; we note that this form of $D(t)$ was first used by \citet{molacek_bush_2013}. We use the same parameters as \citet{PhysRevFluids.7.093604}, reported in Table~\ref{tab:constants}; note that the droplet mass $m$ drops out of our equations with the aforementioned choices of $F(t)$ and $D(t)$.}

\begin{figure*}[t]
    \centering
    \begin{subfigure}{.33\textwidth}
        \caption{$\Gamma/\Gamma_F=0.75$, $3~\cm$ domain}
        \includegraphics[width=.9\linewidth,clip, trim=0 1cm 0 0]{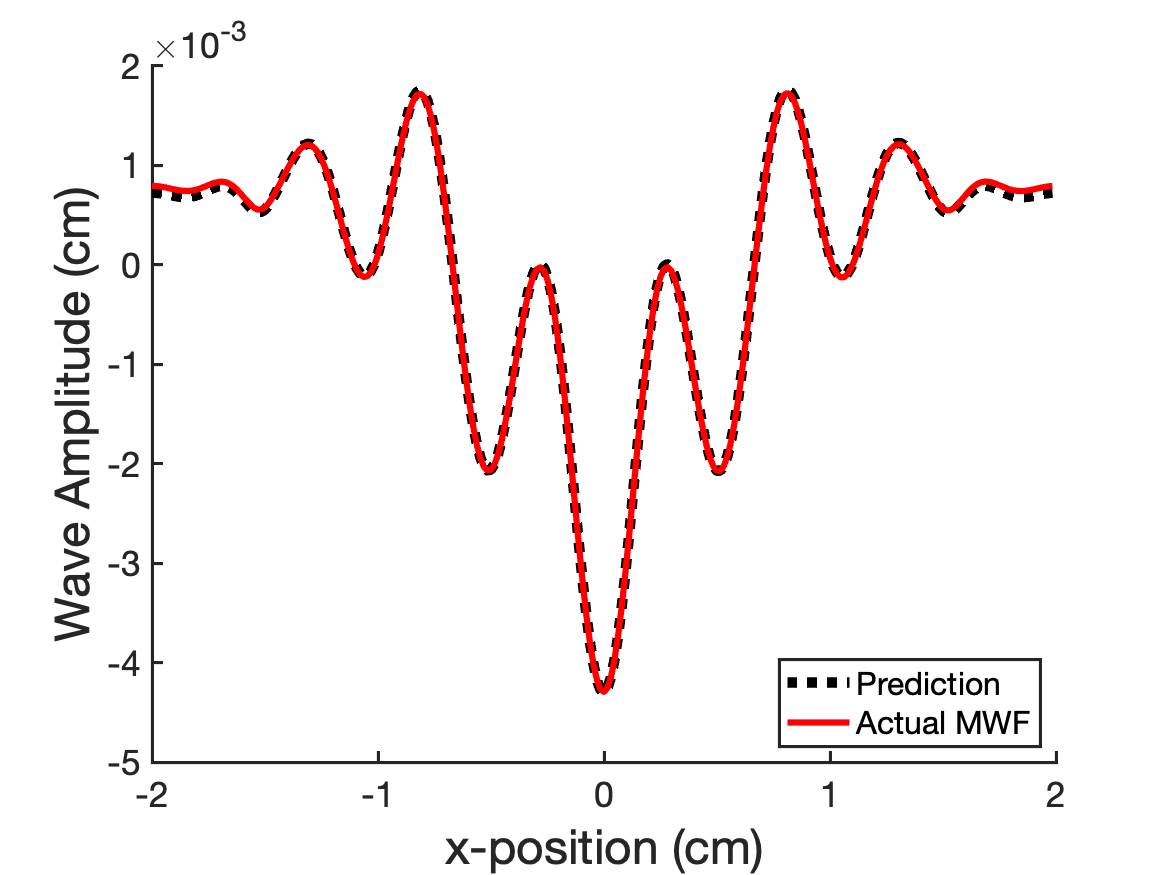}
    \end{subfigure}%
    \begin{subfigure}{.33\textwidth}
        \caption{$\Gamma/\Gamma_F=0.80$, $3~\cm$ domain}
        \includegraphics[width=.9\linewidth,clip, trim=0 1cm 0 0]{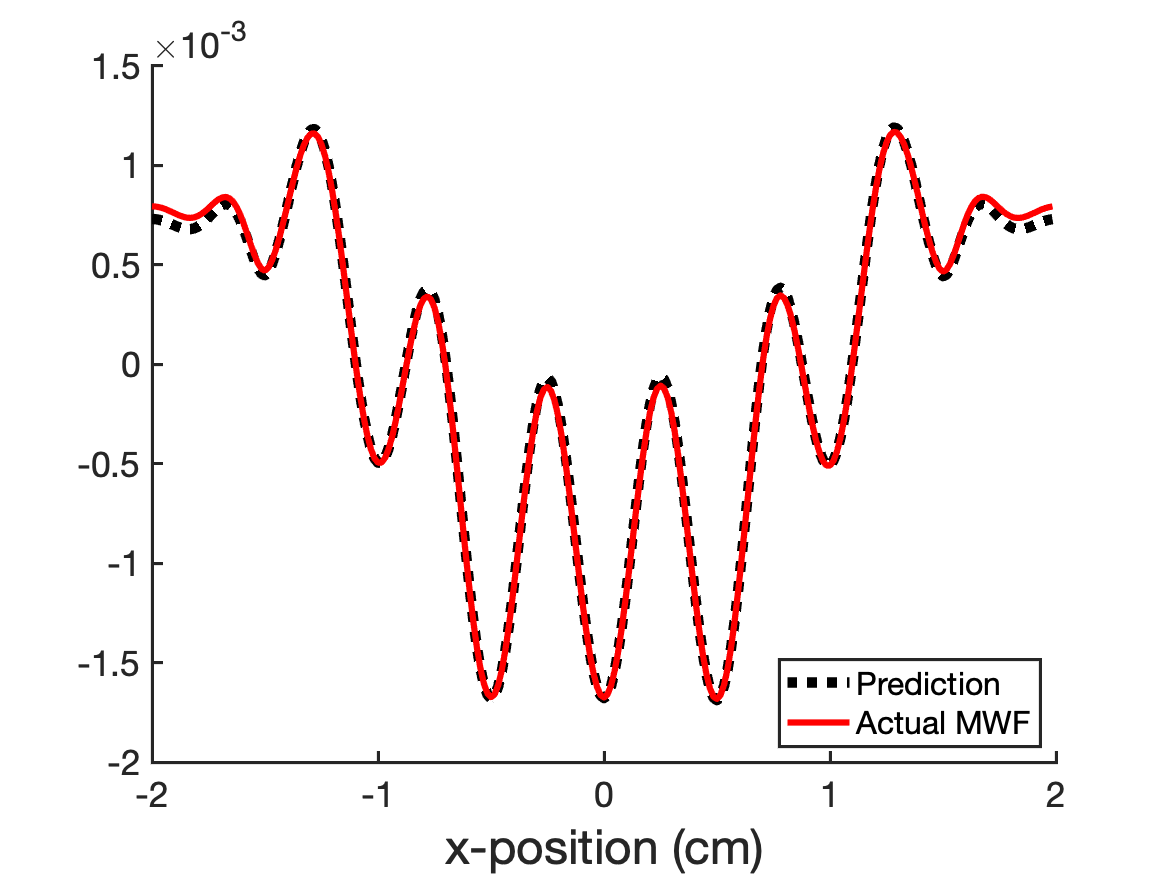}
    \end{subfigure}%
    \begin{subfigure}{.33\textwidth}
        \caption{$\Gamma/\Gamma_F=0.85$, $3~\cm$ domain}
        \includegraphics[width=.9\linewidth,clip, trim=0 1cm 0 0]{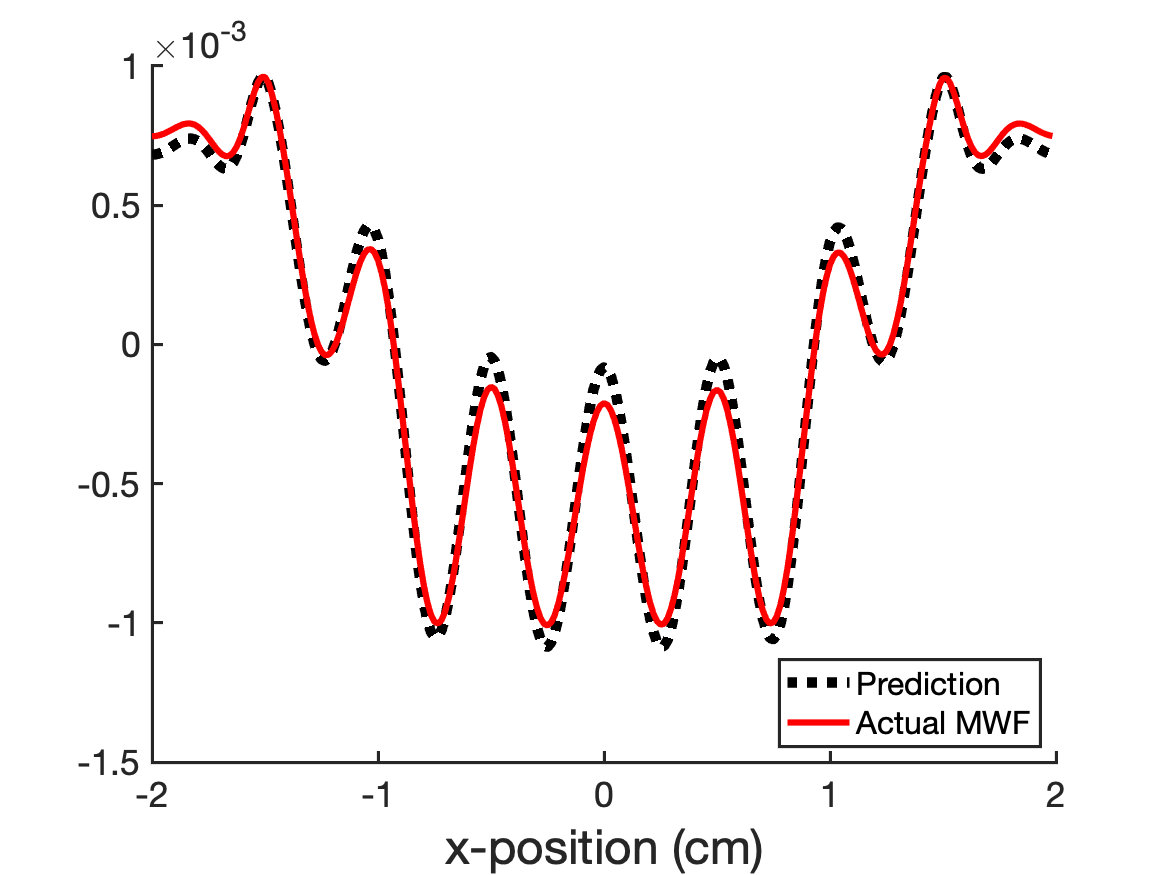}
    \end{subfigure}
    \begin{subfigure}{.33\textwidth}
        \includegraphics[width=.9\linewidth]{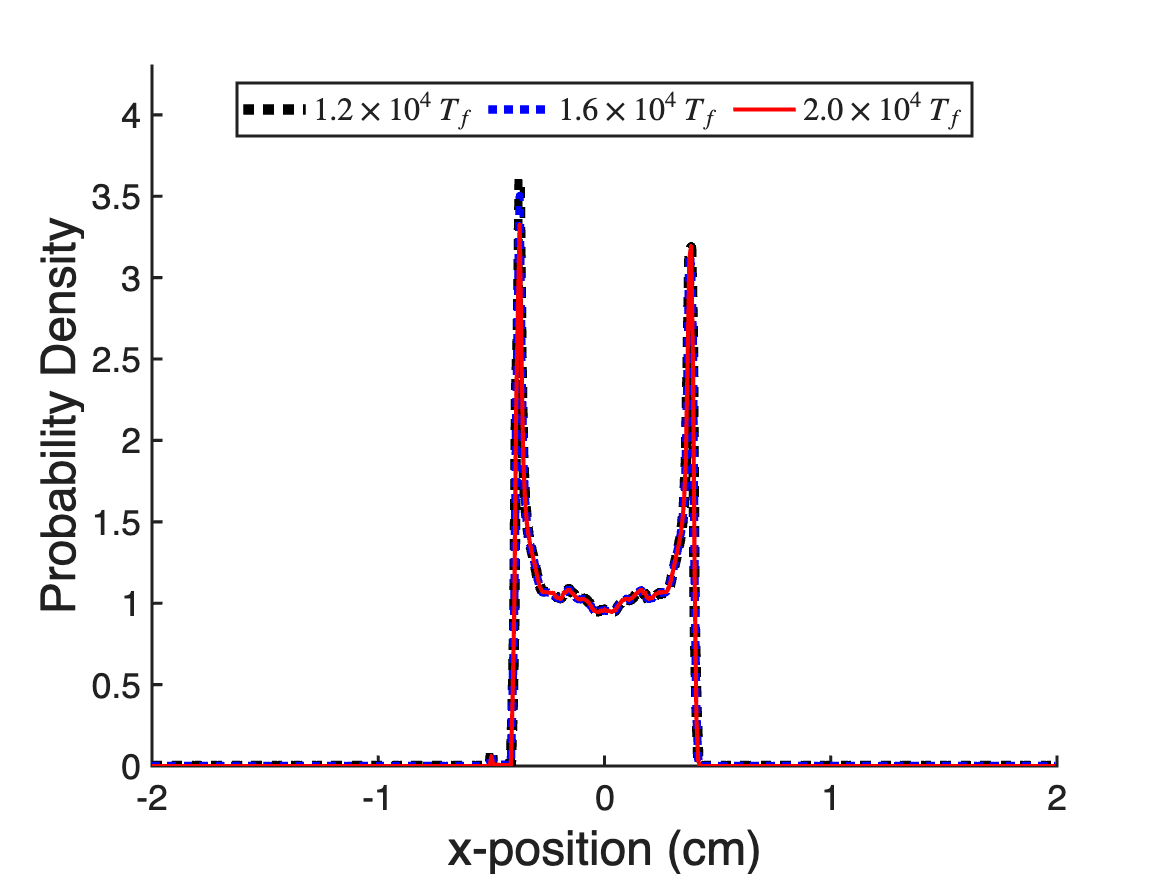}
    \end{subfigure}%
    \begin{subfigure}{.33\textwidth}
        \includegraphics[width=.9\linewidth]{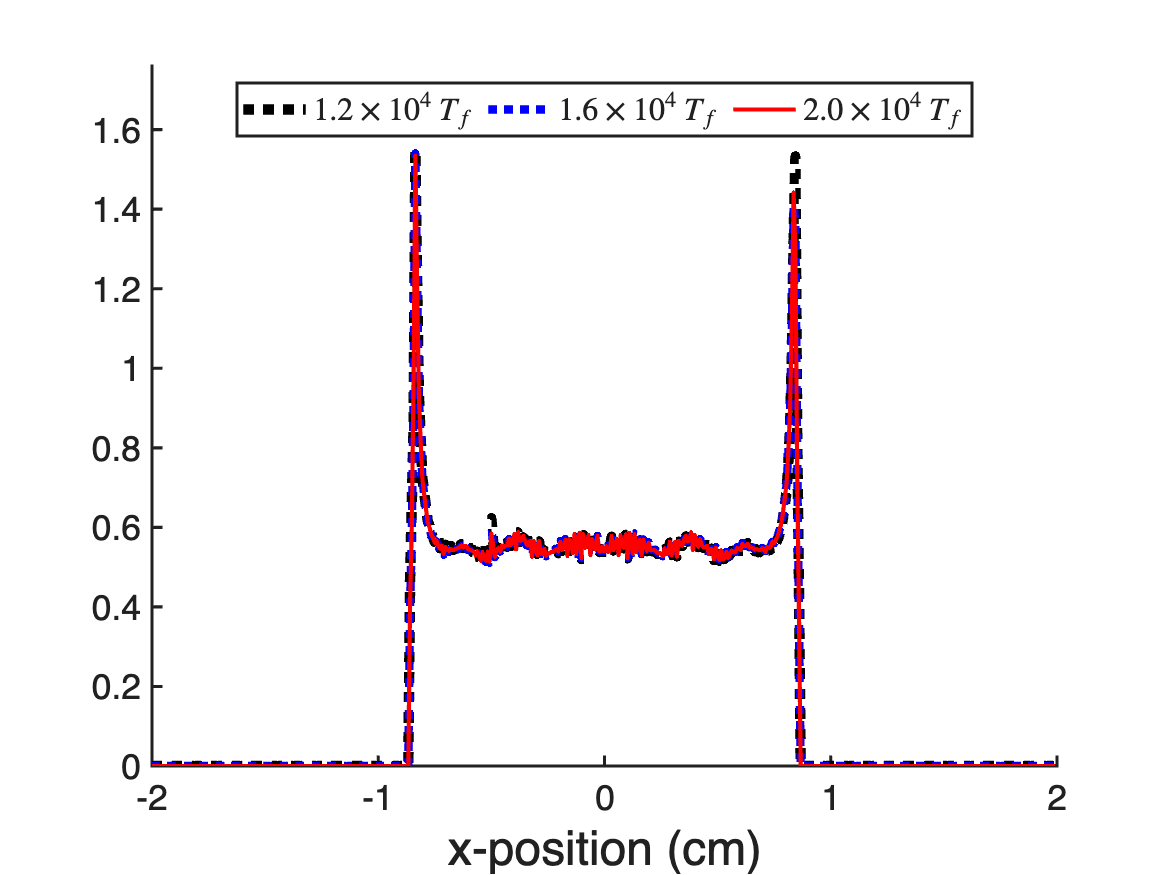}
    \end{subfigure}%
    \begin{subfigure}{.33\textwidth}
        \includegraphics[width=.9\linewidth]{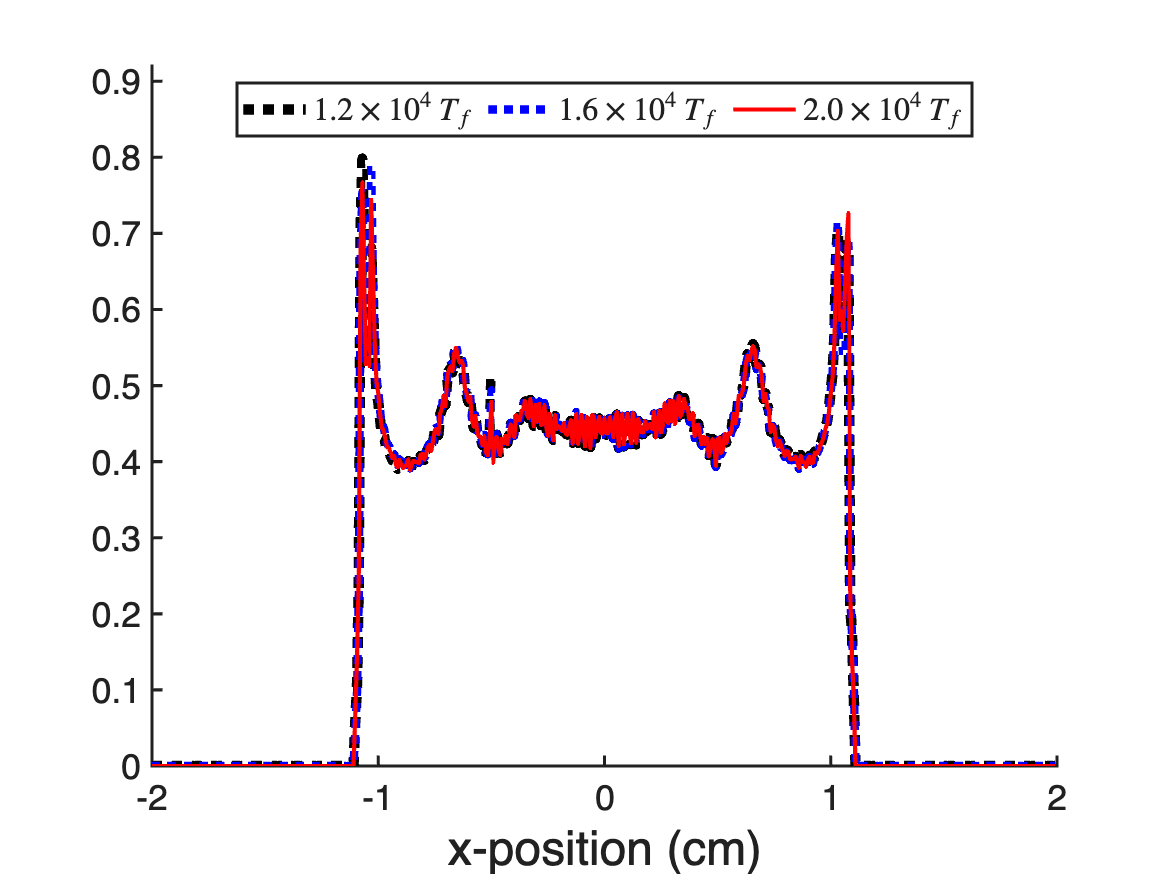}
    \end{subfigure}

    \bigskip\medskip
    
    \begin{subfigure}{.33\textwidth}
        \caption{$\Gamma/\Gamma_F=0.75$, $4~\cm$ domain}
        \includegraphics[width=.9\linewidth,clip, trim=0 1cm 0 0]{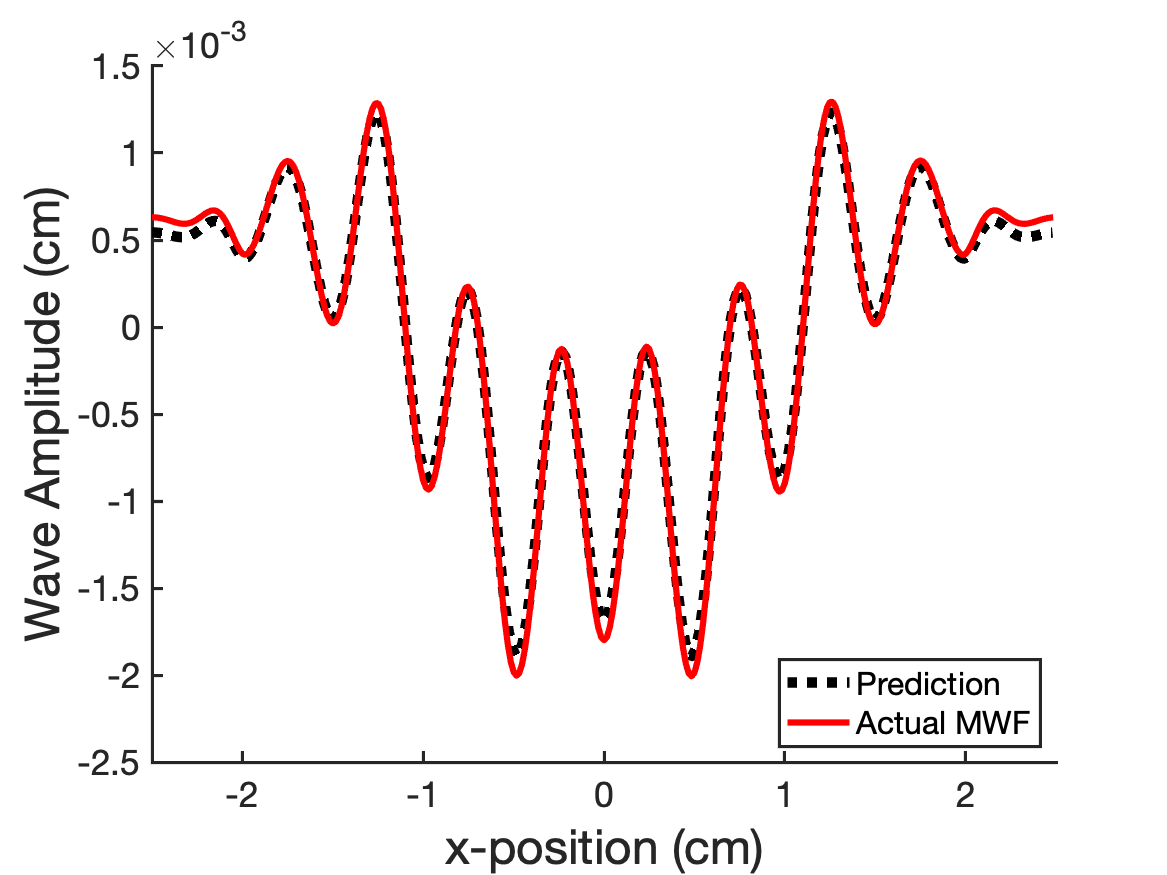}
    \end{subfigure}%
    \begin{subfigure}{.33\textwidth}
        \caption{$\Gamma/\Gamma_F=0.80$, $4~\cm$ domain}
        \includegraphics[width=.9\linewidth,clip, trim=0 1cm 0 0]{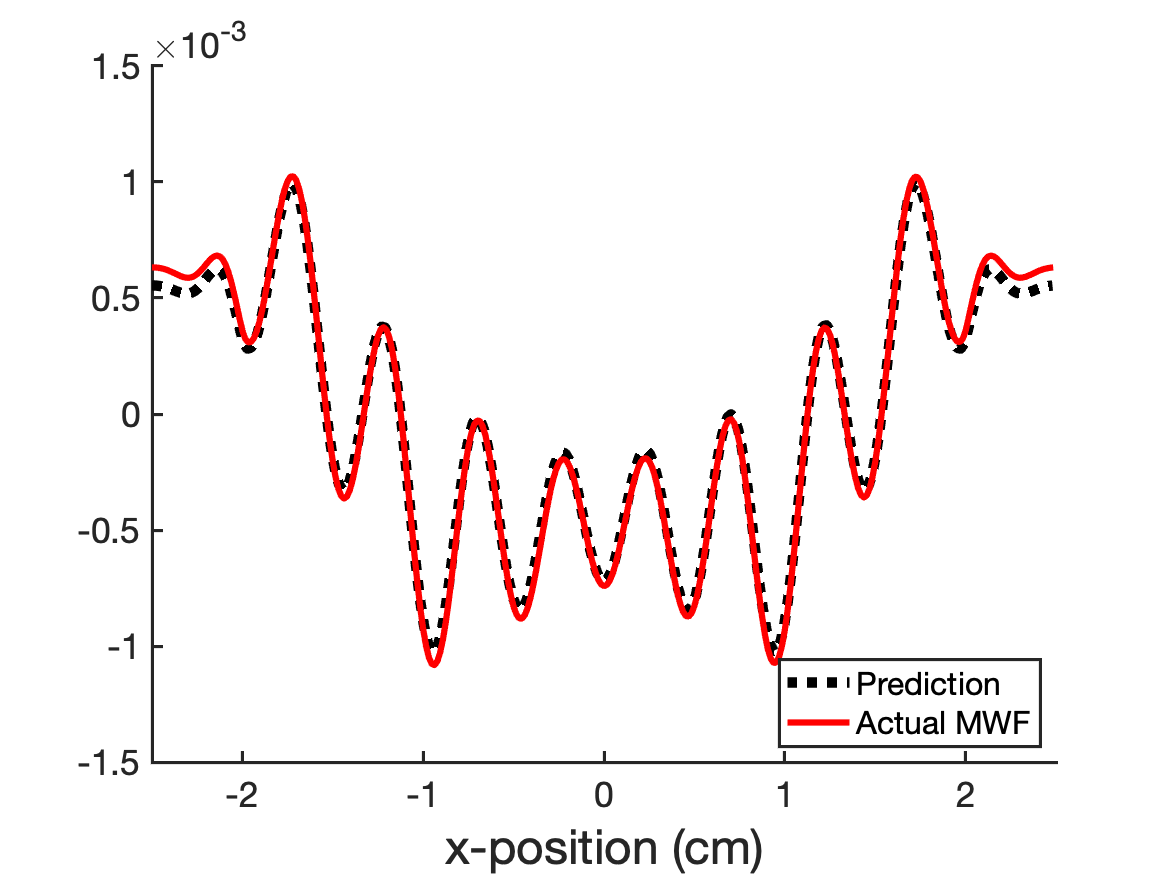}
    \end{subfigure}%
    \begin{subfigure}{.33\textwidth}
        \caption{$\Gamma/\Gamma_F=0.85$, $4~\cm$ domain}
        \includegraphics[width=.9\linewidth,clip, trim=0 1cm 0 0]{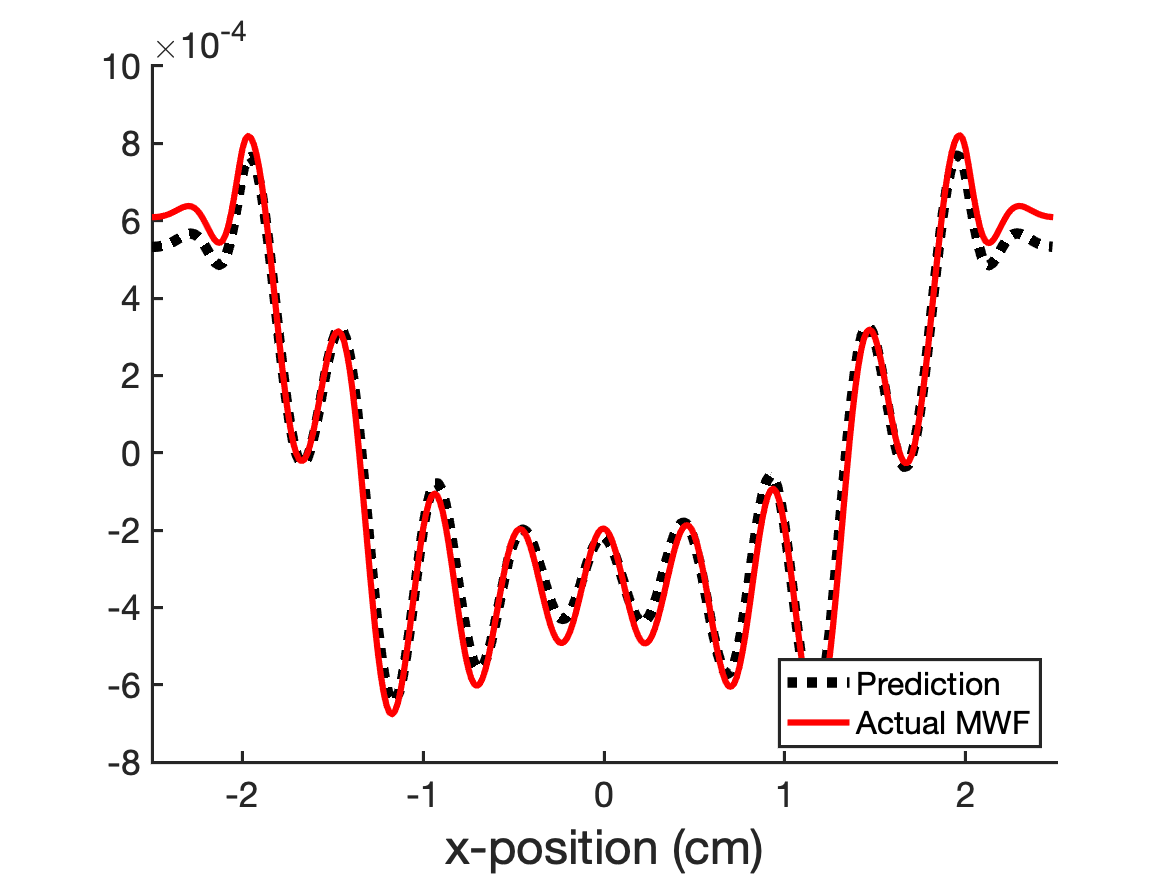}
    \end{subfigure}
    \begin{subfigure}{.33\textwidth}
        \includegraphics[width=.9\linewidth]{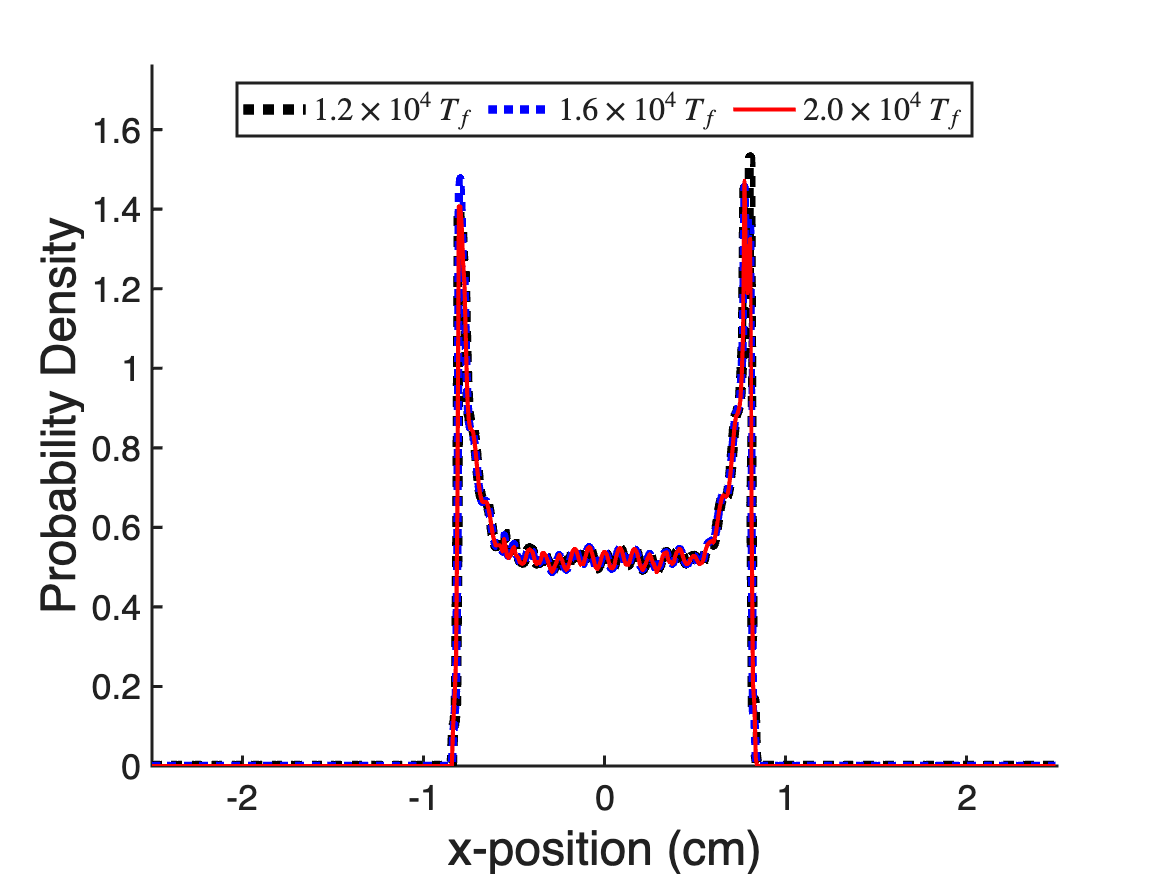}
    \end{subfigure}%
    \begin{subfigure}{.33\textwidth}
        \includegraphics[width=.9\linewidth]{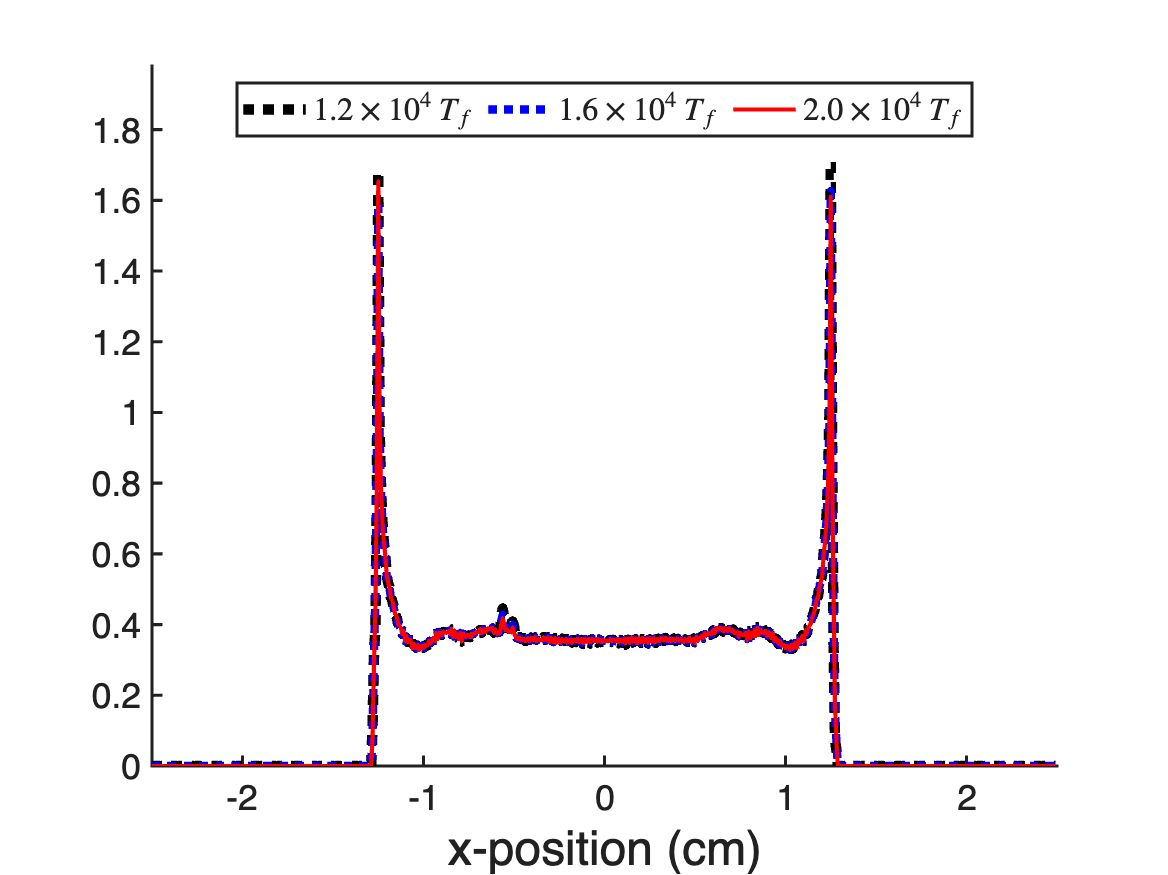}
    \end{subfigure}%
    \begin{subfigure}{.33\textwidth}
        \includegraphics[width=.9\linewidth]{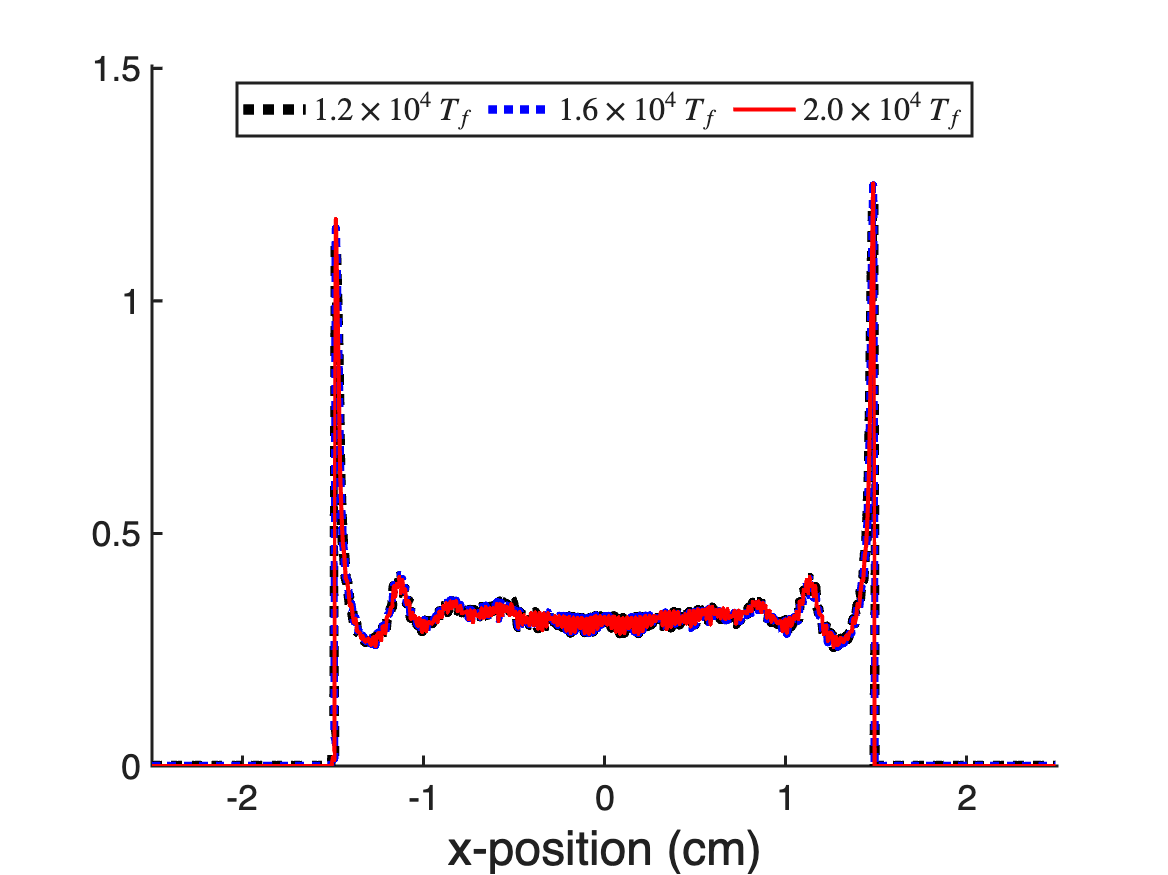}
    \end{subfigure}
    \caption{\dave{Mean wave fields (MWFs) and droplet position PDFs for six simulations, stopped after $2\times 10^4$ Faraday periods $T_F$. In each, a single droplet moves around within a rectangular cavity of width \textbf{(a-c)} $3~\cm$ and \textbf{(d-f)} $4~\cm$, with various memory parameters $\Gamma/\Gamma_F$. For each MWF, we compare the measured MWF against the prediction of Claim~\ref{claim:main_heuristic}, noting a close match in each case. We show each PDF calculated at times $1.2\times 10^4\,T_F$, $1.6\times10^4\,T_F$, and $2.0\times 10^4\,T_F$ to demonstrate convergence.}}
    \label{fig:mwfcomp}
\end{figure*}

\begin{figure*}
    \centering
    \begin{subfigure}{0.33\textwidth}
        \centering
        \includegraphics[width=\linewidth]{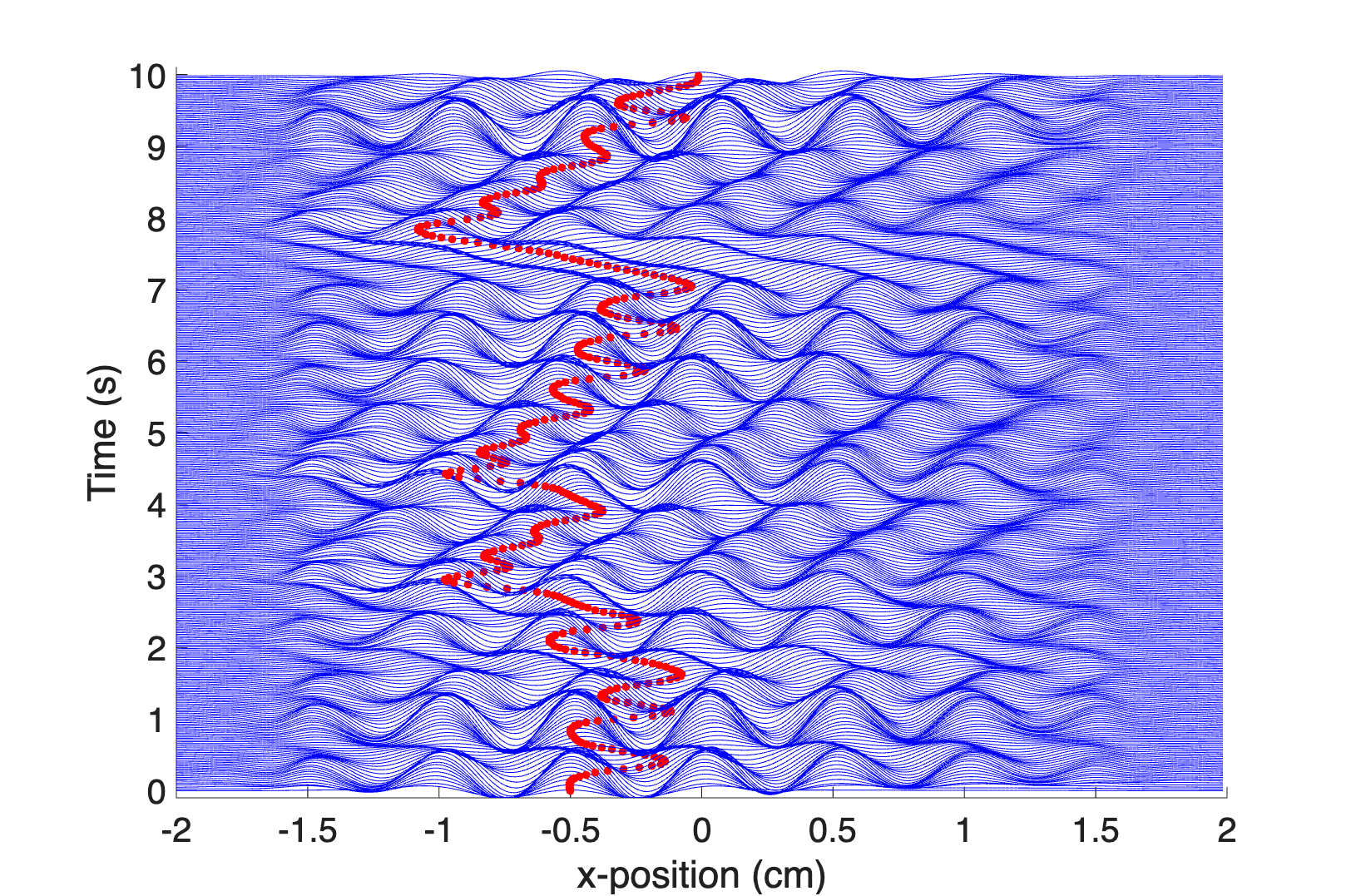}
        \caption{}
        \label{fig:chaostrajectory}
    \end{subfigure}
    \begin{subfigure}{0.33\textwidth}
        \centering\includegraphics[width=\linewidth]{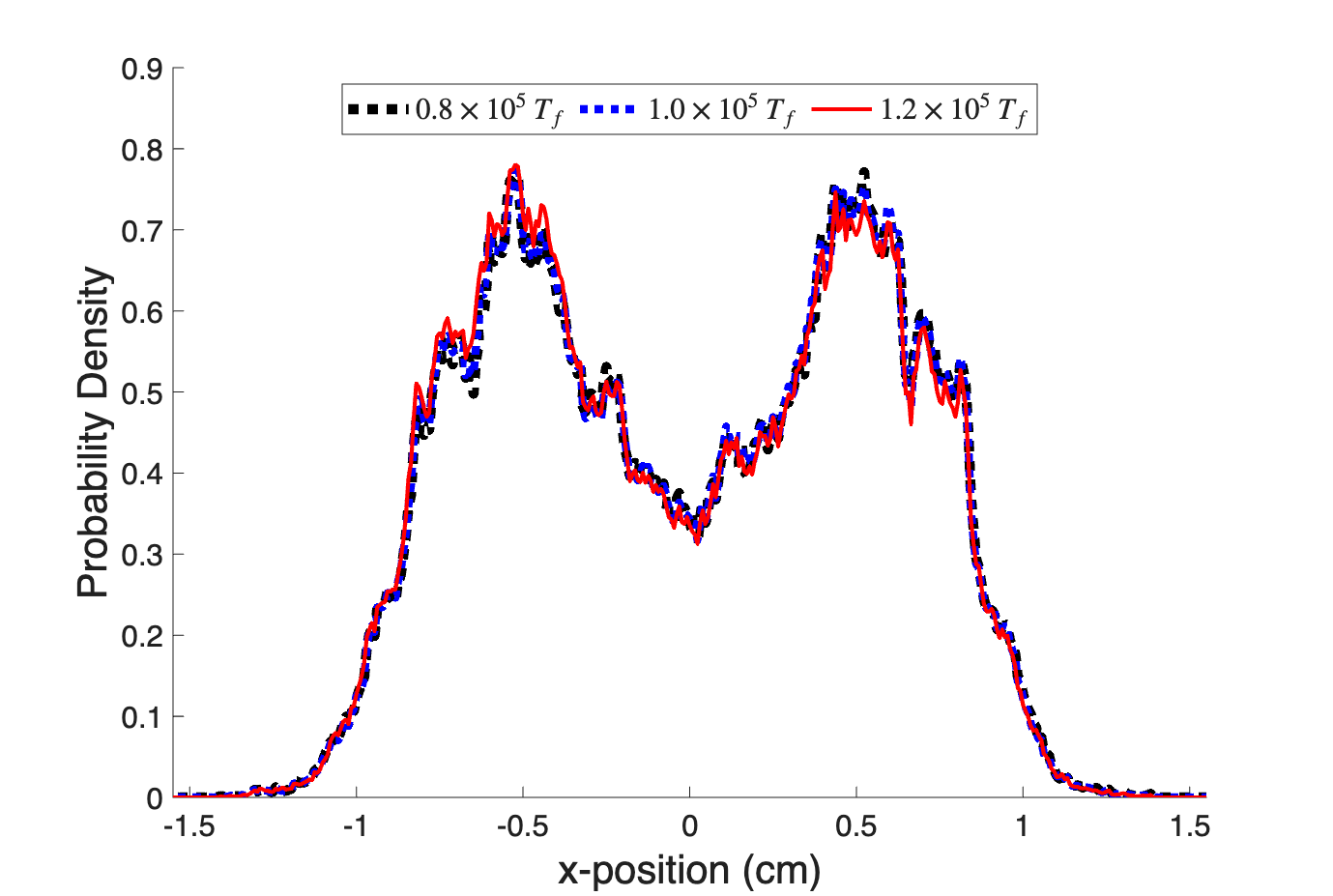}
        \caption{}\label{fig:chaospdf}
    \end{subfigure}%
    \begin{subfigure}{0.33\textwidth}
        \centering\includegraphics[width=\linewidth]{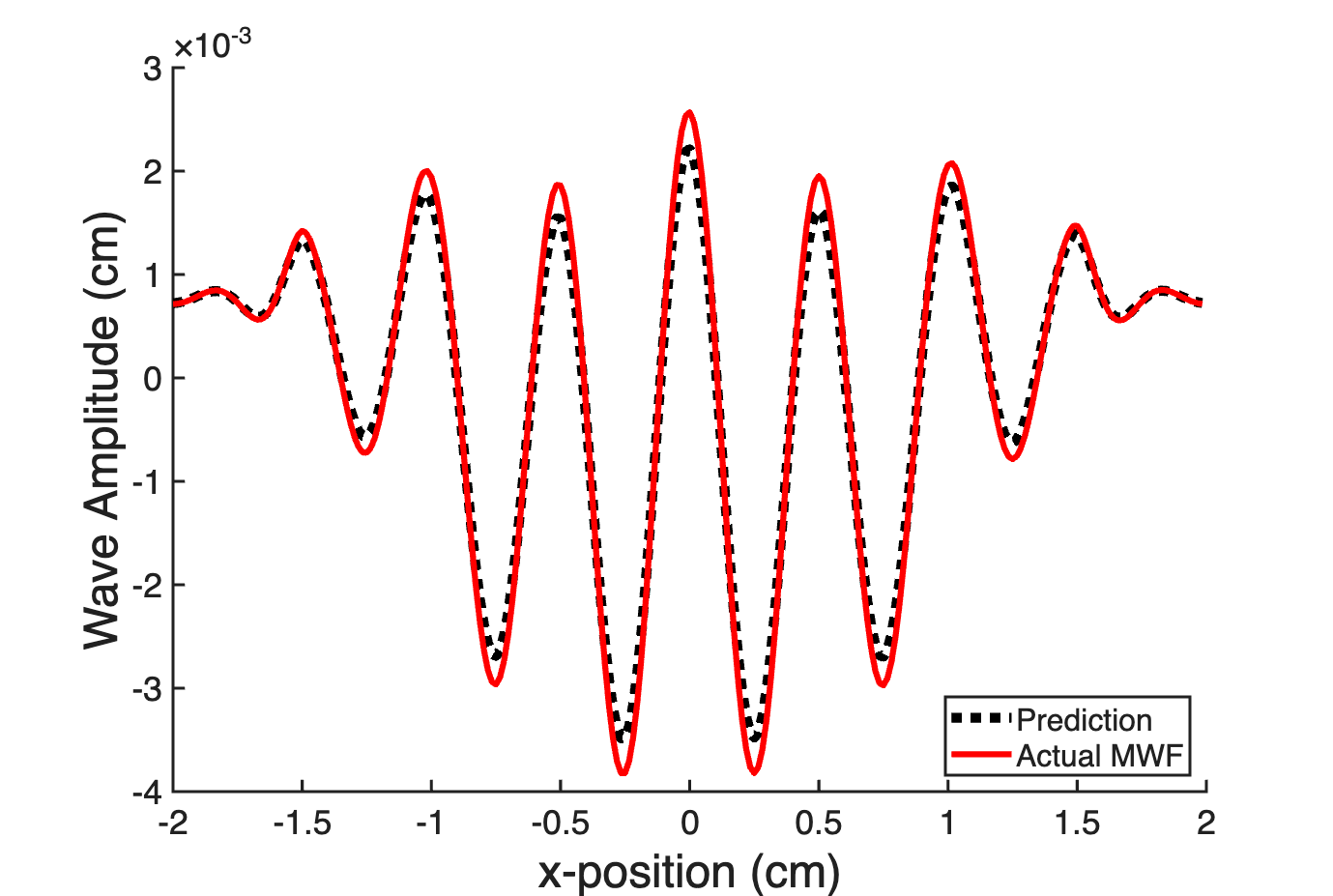}
        \caption{}\label{fig:chaosmwf}
    \end{subfigure}
    \caption{\dave{\textbf{(a)} A droplet (red) starts at $-0.5~\mathrm{cm}$ and traverses the $3~\mathrm{cm}$ domain with memory parameter $\Gamma/\Gamma_F=0.99$. The trajectory of the droplet is chaotic, in contrast to the numerical experiments of Figures~\ref{fig:exdrop} and~\ref{fig:mwfcomp}. \textbf{(b)} The PDF of the droplet's position at times $0.8\times 10^5\,T_F$, $1.0\times 10^5\,T_F$, and $1.2\times 10^5\,T_F$, demonstrating statistical convergence.  \textbf{(c)} The measured MWF at time $1.2\times 10^5\,T_F$ aligns closely with the predicted MWF of Claim~\ref{claim:main_heuristic}.}}
    \label{fig:chaos}
\end{figure*}


We first run simulations on two rectangular domains (width $3~\cm$ and $4~\cm$) with three choices of the memory parameter $\Gamma/\Gamma_F = 0.75$, $0.80$, and $0.85$. In all cases, the depth is $5~\cm$ within the cavity and $0.01~\cm$ outside, and the boundary extends $0.5~\cm$ past the edges of the cavity. The wave field is discretized with 256 horizontal nodes for the $3~\cm$ domain and 384 nodes for the $4~\cm$ domain. We run the simulation with a timestep of $\Delta t=1.25\times 10^{-5}~\s=5\times 10^{-4}\,T_F$. \dave{The droplet starts each experiment at $x_p=-0.5~\cm$ with initial velocity $-0.05~\cm/\mathrm{s}$, to break the left-right symmetry of the system.}

In order to illustrate Claim~\ref{claim:main_heuristic}, we run a long-time simulation ($500~\s$, or $2\times 10^4\,T_F$) of the walking droplet model, recording the instantaneous wave fields and droplet position \dave{once per period, at the start of each bounce} (Figure~\ref{fig:snapshot}). The droplet generates a time-varying wave field as it moves about the cavity (Figure~\ref{fig:trajectory}), and we aggregate these to generate a MWF (Figure~\ref{fig:exdropC}). Separately, we numerically compute the mean wave field $\eta_B(x,y)$ of a stationary bouncer at each point $y$ in the domain, and we integrate the kernel $\eta_B$ against the probability density $\rho(y)$ of our simulation. This calculation yields a prediction for the MWF in accordance with Claim~\ref{claim:main_heuristic}. \dave{In Figure~\ref{fig:mwfcomp}, we compare these predictions against the true MWFs, noting a close match in all cases. The error is higher at higher values of $\Gamma/\Gamma_F$; as $\Gamma$ increases, past values influence the current wave field more strongly, and the MWF converges slower}.

\dave{We note that the droplet dynamics are periodic in the above experiments, with trajectories similar to that of Figure~\ref{fig:trajectory}. In Figure~\ref{fig:chaos}, we run the same numerical experiment with a $3~\cm$ domain and a higher memory parameter $\Gamma/\Gamma_F=0.99$, which now yields \emph{chaotic} droplet dynamics (Figure~\ref{fig:chaostrajectory}). In this system, the statistical convergence of the droplet is slower (Figure~\ref{fig:chaospdf}), so we allow the simulation to run for $3000~\s$ (or $1.2\times 10^5\,T_F$) before stopping. The alignment between the measured and predicted MWFs remains robust (Figure~\ref{fig:chaosmwf}).}

\dave{In Figure~\ref{fig:2drop}, we investigate a two-droplet simulation similar to that of \citet{PhysRevFluids.7.093604}}. Here, we consider the joint PDF $\mu:\mathbb R^2\rightarrow [0,\infty)$ of two droplets in our one-dimensional pilot-wave system. We conduct \dave{simulations} in two different variants of a two-cavity domain, which we refer to as the `uncorrelated' and `correlated' systems, respectively. In both, the left cavity is $1~\cm$ wide and the right cavity is $1.1~\cm$ wide, and both are $0.5~\cm$ deep; breaking the left-right symmetry of the \dave{topography} ensures that the dynamics are not generically symmetric. In the uncorrelated system, the two droplets are separated by a shallow region (of depth $0.01~\cm$), such that their wave fields interact minimally. In the correlated case, the central barrier is replaced by a deep cavity, allowing each droplet to affect the other through their shared wave field. 

We run each simulation for $2\times 10^4\,T_F$, and we report the resulting PDFs, \dave{measured} MWFs, and predicted MWFs in both cases in Figure~\ref{fig:2drop}. As expected, the `uncorrelated' system gives rise to approximately uncorrelated droplet statistics (Figure~\ref{fig:nocorrpdf}), and the resulting MWF is \dave{very closely approximated by the sum} of the two MWFs corresponding to each individual droplet (Figure~\ref{fig:nocorrmwf}). The `correlated' system \dave{yields} highly correlated droplet statistics (Figure~\ref{fig:corrpdf}) and a qualitatively different form of the MWF (Figure~\ref{fig:corrmwf}). The predicted PDF-MWF relationship \dave{approximately holds in both cases}. 

\begin{figure*}
    \begin{subfigure}[b]{0.33\textwidth}
        \includegraphics[width=\linewidth]{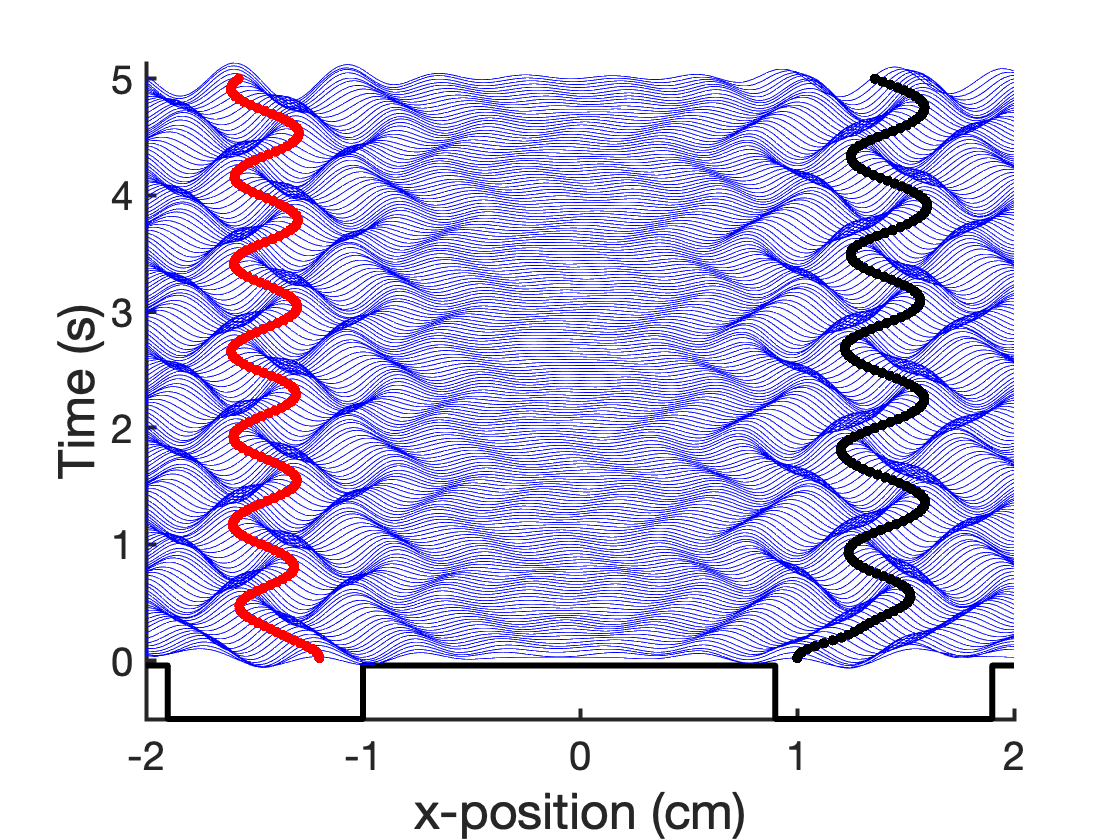}
        \caption{Uncorrelated Trajectories}\label{}
    \end{subfigure}%
    \begin{subfigure}[b]{0.33\textwidth}
        \includegraphics[width=\linewidth]{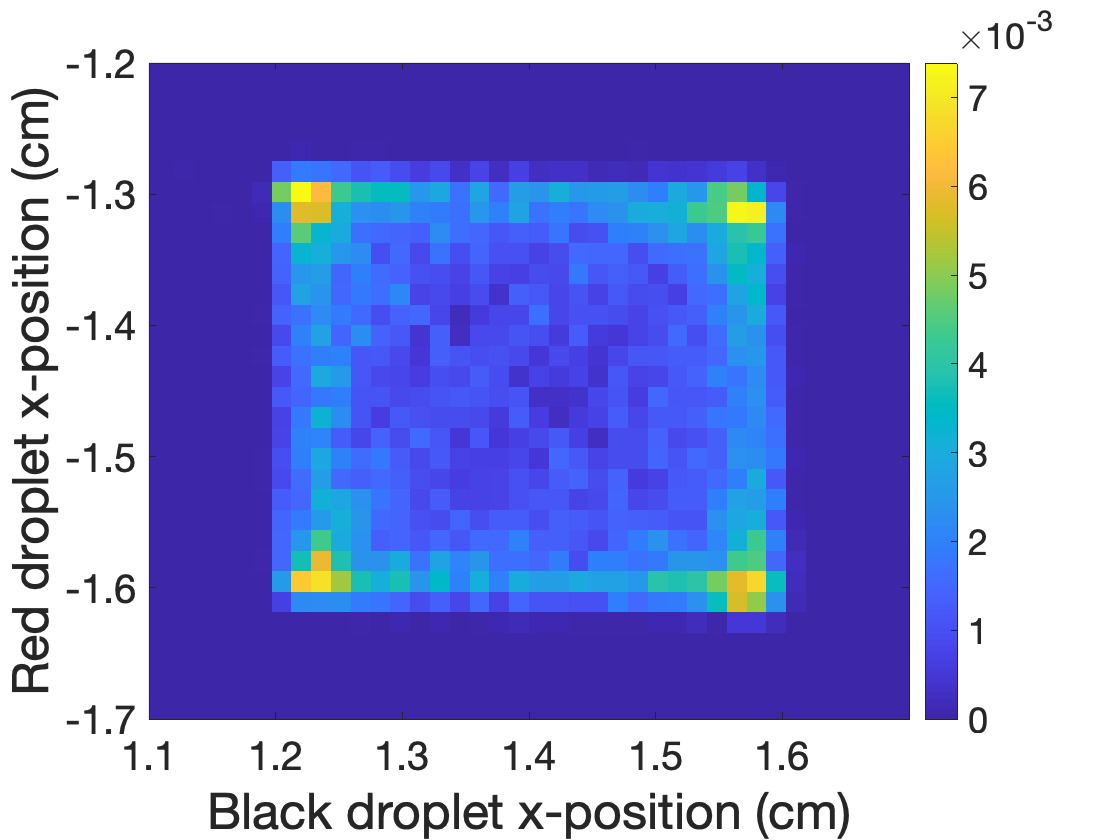}
        \caption{Uncorrelated PDF}\label{fig:nocorrpdf}
    \end{subfigure}%
    \begin{subfigure}[b]{0.33\textwidth}
        \includegraphics[width=\linewidth]{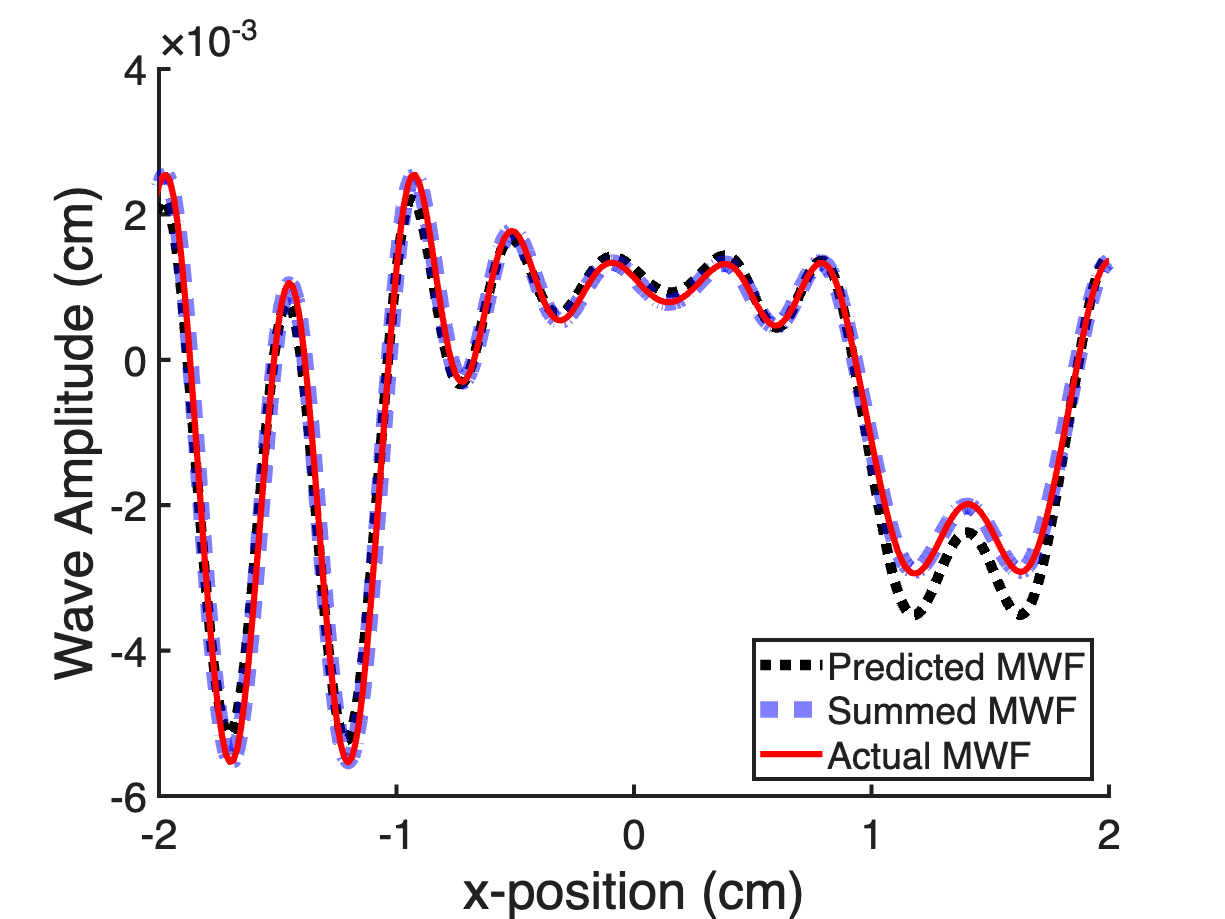}
        \caption{Uncorrelated MWF}\label{fig:nocorrmwf}
    \end{subfigure}
    \centering
    \begin{subfigure}[b]{0.33\textwidth}
        \includegraphics[width=\linewidth]{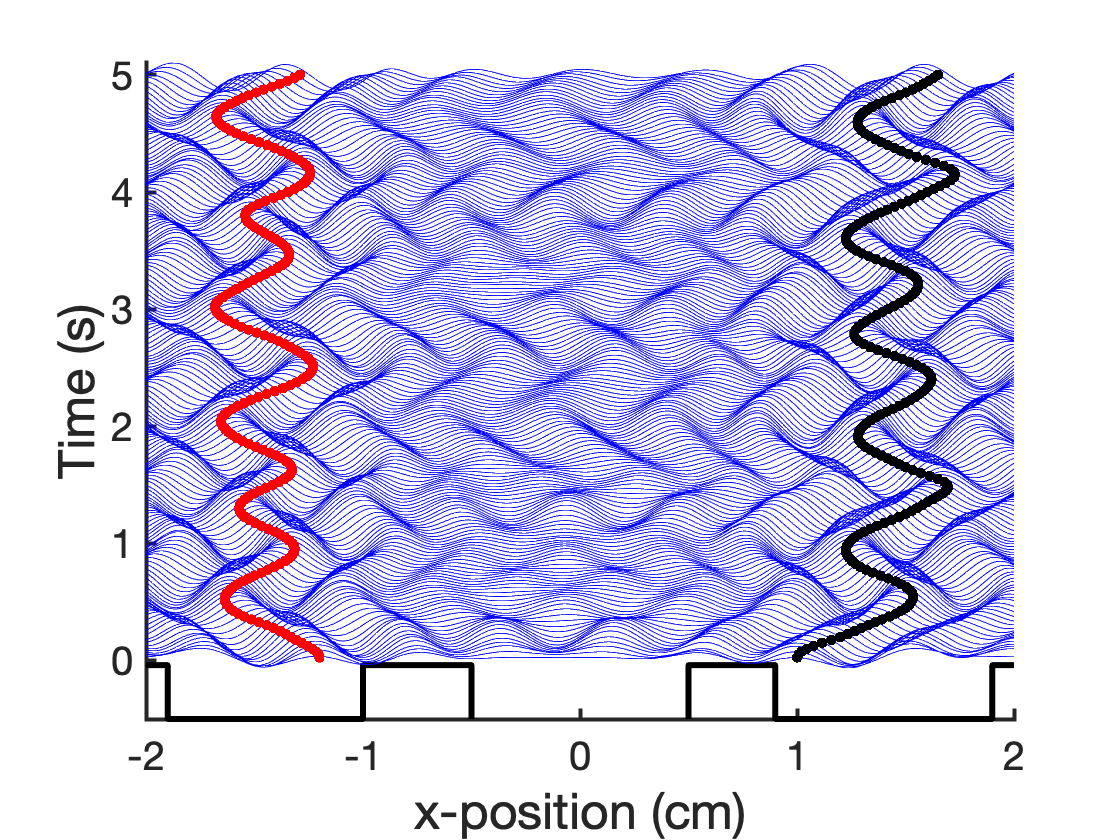}
        \caption{Correlated Trajectories}\label{}
    \end{subfigure}%
    \begin{subfigure}[b]{0.33\textwidth}
        \includegraphics[width=\linewidth]{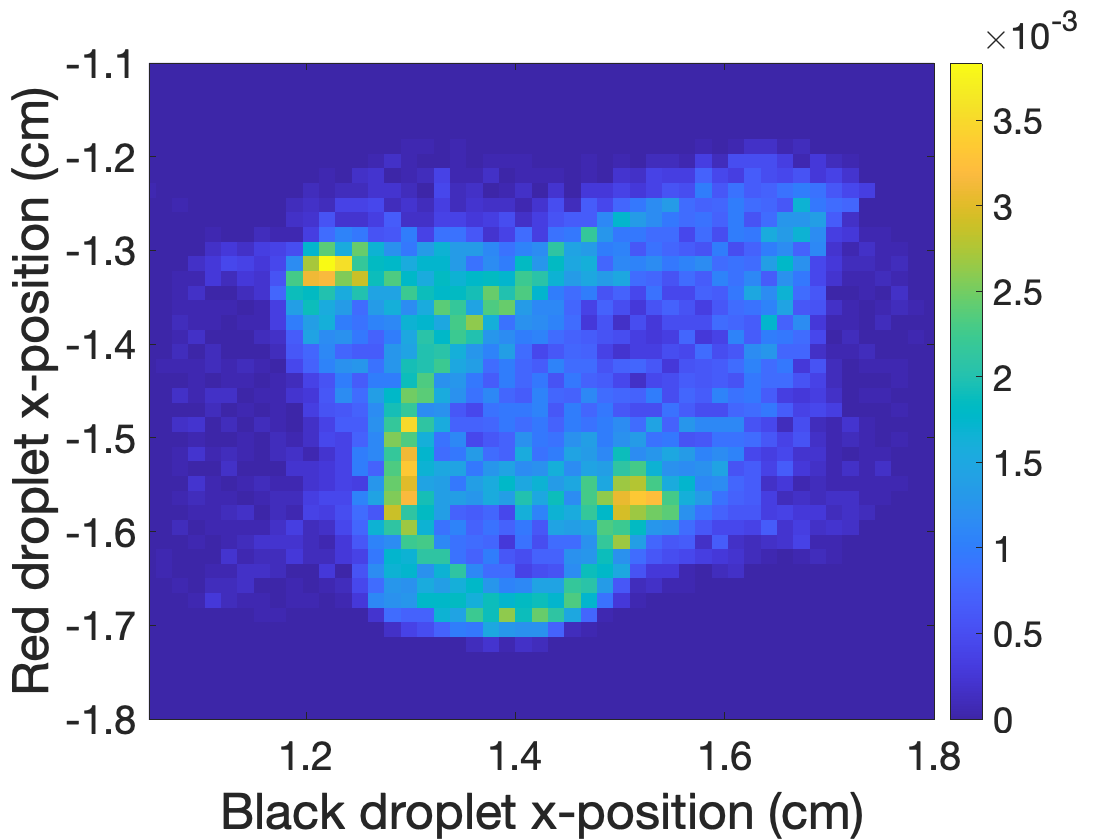}
        \caption{Correlated PDF}\label{fig:corrpdf}
    \end{subfigure}%
    \begin{subfigure}[b]{0.33\textwidth}
        \includegraphics[width=\linewidth]{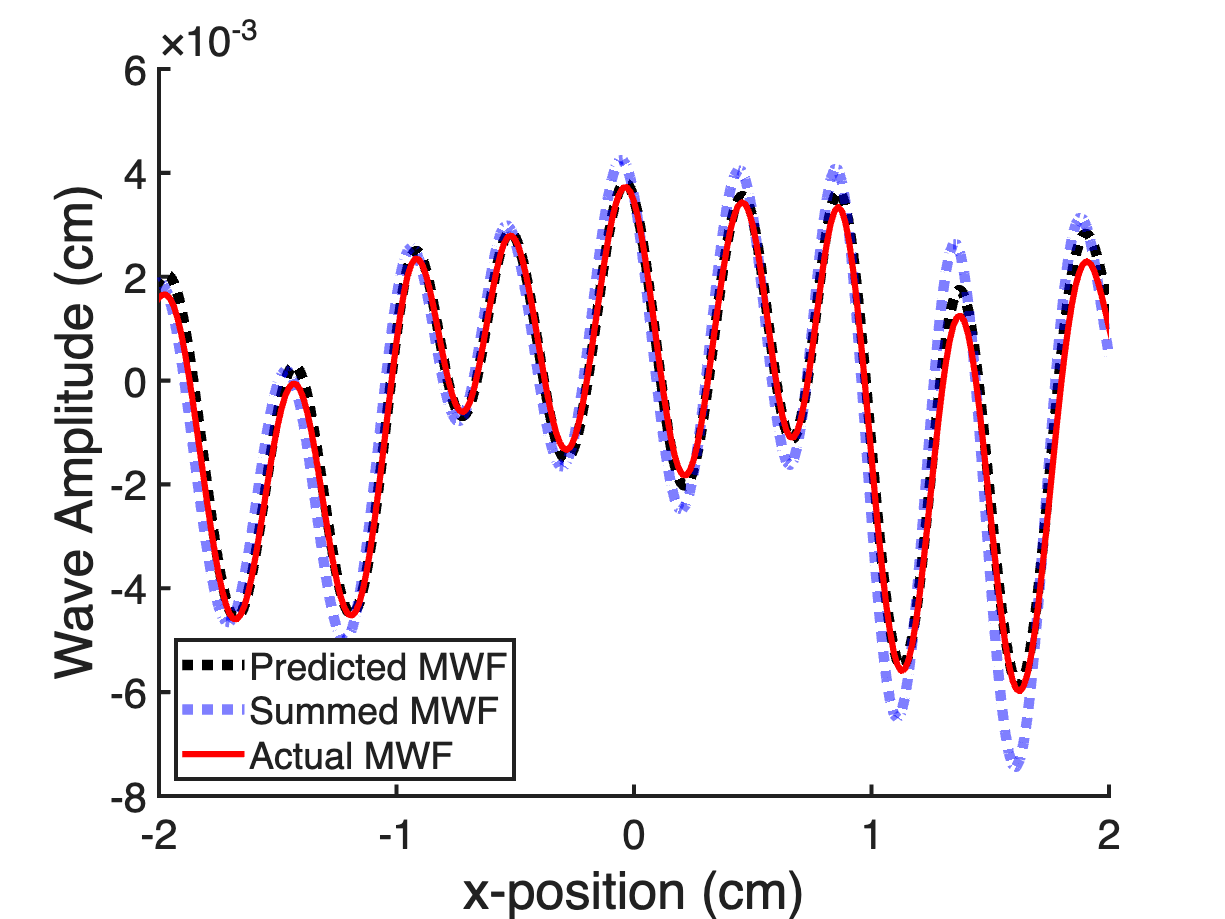}
        \caption{Correlated MWF}\label{fig:corrmwf}
    \end{subfigure}
    \centering
    
    \caption{Two droplets move in each of the asymmetrical two-cavity systems, which we label `uncorrelated' and `correlated', respectively. The uncorrelated domain consists of two cavities separated by a wide central barrier, preventing the droplets from communicating with one another; the correlated domain replaces the barrier with another cavity, such that the droplets interact through their shared wave field. One droplet starts at $-1.20~\cm$ (red) and another starts at $1.00~\cm$ (black). \textbf{(a,d)} The uncorrelated and correlated systems, respectively, during the first 200 Faraday periods. \textbf{(b,e)} PDFs of the two systems after $1.8\times 10^4$ Faraday periods. \textbf{(c,f)} \dave{The comparison between the measured MWFs, the predicted MWFs, and the sum of the two MWFs that arise when one of the two droplets is removed from the system.}}
    \label{fig:2drop}
\end{figure*}

\section{Applications of PDF-MWF Correspondence}\label{sec:apps}
In the present section, we highlight various applications of our results to current themes in walking droplet research, including interactions between multiple droplets, finely-resolved models of the droplet-bath interaction, and droplets bouncing out-of-phase with their guiding wave.

\subsection{Multi-Droplet Statistics}
Claim~\ref{claim:main_heuristic} encompasses systems with several droplets, even if they do not impact the bath simultaneously. \dave{For a simple example, suppose we have $N$ droplets, each of which impacts the bath once per vibration period, and suppose the impact locations\footnote{\dave{In practice, the `impact location' must be chosen by averaging the droplet's position over a small time interval. Our argument is agnostic to the specifics of this averaging process, as long as the resulting positions provide sufficient data to accurately characterize the droplet-bath interaction in a given period.}} $x_1,...,x_N\in\Omega$ are sufficient to characterize the droplet-bath interaction in a given period. If these positions are distributed} according to a stationary distribution $\rho = \rho(x_1,...,x_N)$, the mean wave field $\ovl{\eta}$ is given by
\[\ovl{\eta}(x) = \int_{\Omega^N}\eta_B(x,x_1,...,x_N)\rho(x_1,...,x_N)\,dx_1\cdots dx_N,\]
where $\eta_B(x,x_1,...,x_N)$ is the mean wave field corresponding to $N$ stationary bouncers at $(x_1,...,x_N)$. Investigations of multi-droplet interactions have uncovered a variety of unique behaviors, including orbiting and promenading bound states~\citep{10.1063/1.5032116,10.1063/1.5032114,Couchman_Turton_Bush_2019}, regular lattices~\citep{Eddi_2009,sym14081524} and ring structures~\citep{osti_10201541}, and long-distance correlations between droplets~\citep{PhysRevFluids.7.093604}.

\subsection{Arbitrary Droplet-Bath Impact Models}
In another direction, Claim~\ref{claim:main_heuristic} places few restrictions on how one models the droplet-bath interaction. Although phenomenological and reduced models of walking droplets have employed instantaneous droplet-bath impacts~\citep{fort_2010,Durey_Milewski_2017}---such as the \emph{stroboscopic} model of \citet{Oza_Rosales_Bush_2013}, which eliminates vertical dynamics and phase from consideration---more complete models account for the spring-like interaction between the droplet and the wave~\citep{molacek_bush_2013,milewski_galeano-rios_nachbin_bush_2015}. With these dynamics in place, such models are able to capture important nuances missed in earlier work, such as modulations of the vertical bouncing state~\citep{doi:10.1063/1.5032221}. Claim~\ref{claim:main_heuristic} can be applied directly to such models, \dave{and indeed, our numerical results in Section~\ref{sec:numerics} make use of a finite droplet-bath contact time}.

\subsection{Non-Resonant Walking Droplets}
Finally, \citet{Primkulov_2025} and \citet{Evans_2025} have recently undertaken a study of droplets bouncing \emph{out of resonance} with the underlying bath. They found that non-resonant bouncing helps to rationalize a number of subtle details of the droplet system, including the `backtracking' behavior reported by \citet{PhysRevLett.117.094502} and the intermittent mode switching seen in high memory systems. 

There are two ways to incorporate such effects into the present framework. First, if an instantaneous pressure model suffices, one can model a single droplet using its position $x$ and its \dave{instantaneous} impact phase $\theta$. If both are distributed according to a stationary distribution $\rho = \rho(x,\theta)$, then the mean wave field $\ovl{\eta}$ is given by
\[\ovl{\eta}(x) = \int_{\Omega\times S^1} \eta_B(x,x',\theta')\rho(x',\theta')\,dx'd\theta',\]
where $\eta_B(x,x',\theta')$ is the mean wave field of a bouncer with a consistent impact phase. If an instantaneous pressure model does \emph{not} suffice, we must account for the fact that a single period $[0,T]$ of wave oscillation could contain effects from \emph{two} droplet bounces: the end of the previous bounce (which started at $t<0$) and the start of a new bounce (at $t\in[0,T]$). We could thus model a droplet as carrying two impact phases\footnote{\dave{There are various conventions for defining the `impact phase' when the droplet-bath contact time is finite; for instance, \citet{molacek_bush_2013} use a mean impact phase, weighted by the bath's oscillation. So long as the prescription for $\theta$ is consistent (and the resulting impact phase is sufficient to approximately characterize the droplet-bath impact), our results hold as stated.}}, $\theta_1$ and $\theta_2$, and proceed as before. 

With either approach, we expect that Claim~\ref{claim:main_heuristic} should recover the mean wave field behavior seen in the work of \citet{Evans_2025}. In particular, in conditions where the droplet undergoes frequent mode switching---i.e., the impact phase $\theta$ alternates between $\pm\theta_0$---they observe that $\ovl{\eta}\to 0$ over time. In this case, the force on the droplet generated by the mean wave field vanishes, leaving only a `ponderomotive' force generated by the root-mean-squared (RMS) value 
\[\langle\eta(x)^2\rangle^{1/2} = \left(\frac{1}{N}\sum_{n=1}^N \eta(x,nT_F)^2\right)^{1/2}\]
of the wave field. We believe that the RMS wave field can be studied in a general system using techniques \dave{similar to} those presented here, but further investigation of the RMS wave field is outside the present scope.

\section{Conclusion}
In this work, we have rigorously established a far-reaching correspondence between the statistics of a walking droplet and the time-average of its guiding wave (Claim~\ref{claim:main_heuristic}). Our work extends the results of \citet{durey_chaos}, \citet{Durey_Milewski_Wang_2020}, and \citet{Abraham_2024} in the case of single-droplet systems with instantaneous droplet-bath contact time, and it corroborates the numerical and experimental findings of \citet{Saenz2018}, \citet{kutz_pilot-wave_2023}, and \citet{Evans_2025}. \dave{It also greatly strengthens the analogy between the statistics of walking droplets and those of quantum particles. In particular, Born's rule~\eqref{eq:born} applies equally well when a quantum state $|\vec{\alpha}\rangle$ in a (Hilbert) state space $\mathcal{H}$ is parameterized by multiple values $\vec{\alpha}=\{\alpha_1,...,\alpha_N\}$, which may include particle positions, momenta, spins, or energy values; in this case, Born's rule reads~\cite{griffiths2017introduction}
\[\rho(\vec{\alpha}) = \big|\langle \vec{\alpha}| \psi\rangle\big|^2,\]
writing $|\psi\rangle\in\mathcal{H}$ for the particle's state and $\langle\,\cdot\,|\,\cdot\,\rangle$ for the inner product on $\mathcal{H}$. As one consequence of our work, we see that a similar principle applies to the correspondence between droplet statistics and the mean wave field.}

The present approach has promising applications to walking droplet statistics more broadly. In one direction, our numerical simulations in Section~\ref{sec:numerics} highlight how the mean wave field can be used to diagnose correlation between distinct droplets. In another direction, \citet{Evans_2025} have observed that the magnitude of the mean wave field indicates the relative importance of second-order, `ponderomotive' forces on the motion of droplets; using the results presented here, one should be able to estimate the importance of ponderomotive effects \emph{a priori}. Going further, \dave{as discussed in Section~\ref{sec:apps}}, we believe our analysis can be adapted to show how these ponderomotive effects---which depend on second moments of the wave field, rather than simply its mean---might be recovered from droplet statistics alone.

Finally, we note that our analysis can be applied similarly to any time-periodic driven-dissipative PDEs, such as geophysical fluid flows under tidal forcing or ecological systems experiencing daily or seasonal variations. For such systems, Claim~\ref{claim:main_heuristic} should allow one to invoke a similar `PDF-MWF correspondence' to study the long-time response to ergodic or periodic forcing.

\begin{acknowledgments}
We would like to thank the organizers of MIT PRIMES for connecting the two authors, and for their constructive suggestions on the research and manuscript. We would also like to thank John Bush for helpful discussions on pilot-wave hydrodynamics and Glenn Flierl for his guidance in preparing the present manuscript. The second author would like to acknowledge the support of an NDSEG Graduate Fellowship.
\end{acknowledgments}

\section*{Data Availability Statement}
The data that support the findings of this study are available from the corresponding author upon request.

\section*{Conflict of Interest Statement}
The authors have no conflicts to disclose.

\appendix
\section{Statement and Proof of Claim~\ref{claim:main_heuristic}}\label{sec:main}
We prove Claim~\ref{claim:main_heuristic} in the present section. In order to clarify our method of approach, we first consider the case of a time-periodic, finite-dimensional ODE, \dave{which can be handled using classical Floquet theory. Physically, the finite-dimensional case can be taken to represent a fixed numerical discretization of our infinite-dimensional wave equation. For more detail on Floquet theory and its applications to PDEs, we recommend the work of \citet{Kuchment1993}.} 

\subsection{Classical Floquet Theory for the Finite-Dimensional Case}\label{sec:finitedim}
As a starting point, for a fixed period $T>0$, consider the differential equation
\begin{equation} 
\dot{y}=L(t)y+b(t),\label{eq:floquetdrop}
\end{equation}
where $L(t)\in\CC^{d\times d}$ and $b(t)\in\CC^d$ are continuous and $T$-periodic. The following result shows how to apply Floquet's theory to an inhomogeneous equation, and it establishes the pattern that we will later replicate in the PDE setting:

\begin{lemma}\label{lem:floquetsol}
    Suppose $L(t)$ and $b(t)$ are $T$-periodic and continuous. By Floquet's theorem \citep{floquet}, the solutions to $\dot{y}=L(t)y$ can be organized into a fundamental matrix $Y(t)\in\CC^{d\times d}$ that takes the form $Y(t) = Q(t)e^{Bt}$, where $B\in\CC^{d\times d}$ is constant and $Q(t)\in\CC^{d\times d}$ is $T$-periodic, continuously differentiable, and invertible for all time.

    With this notation in place, any solution $y(t)$ to the inhomogeneous equation~\eqref{eq:floquetdrop} satisfies
    \begin{equation}\label{eq:difference_eq}
        y((n+1)T+t_0)=C(t_0)y(nT+t_0)+R(t_0),
    \end{equation}
    for any non-negative integer $n$ and any time offset $t_0\in[0,T]$; here, the quantities $C(t_0)=Q(t_0)e^{BT}Q^{-1}(t_0)$ and
    \begin{equation}
        R(t_0)=e^{B(t_0+T)}\int_{t_0}^{t_0+T}e^{-Bs}Q^{-1}(s)b(s)~ds\label{eq:r0}
    \end{equation}
    are dependent only on $t_0$.
\end{lemma}
\begin{proof}
Without loss of generality, suppose $Q(0)=1$. Substituting the expression $Y(t)=Q(t)e^{Bt}$ into $\dot{Y}(t) = L(t)Y(t)$ and removing common factors, we have
\begin{equation}
\dot{Q}+QB=LQ. \label{eq:floquetsub}
\end{equation}
Separately, let $z(t)=Q^{-1}(t)y(t)$ for a given solution $y(t)$ to~\eqref{eq:floquetdrop}. We have
\begin{equation}
\dot{Q}z+Q\dot{z}=LQz+b.\label{eq:qzsub}
\end{equation}
After right multiplying~\eqref{eq:floquetsub} by $z$, we combine~\eqref{eq:floquetsub} and~\eqref{eq:qzsub} to obtain
\begin{equation}
\dot{z}=Bz+Q^{-1}b,
\end{equation}
which is solved by
\[z(t)=e^{Bt}z_0+e^{Bt}\int_0^t e^{-Bs}Q^{-1}(s)b(s)~ds,\]
for a given $z_0=z(0) = y(0)$. Setting $t = nT+t_0$, we find
\begin{align*}
z(t+T)&=e^{B(t+T)}z_0+e^{B(t+T)}\int_0^{t+T}e^{-Bs}Q^{-1}(s)b(s)~ds\\
&= e^{B(t+T)}z_0+e^{B(t+T)}(e^{-Bt}z(t)-z_0)\\
&\qquad+e^{B(t+T)}\int_t^{t+T}e^{-Bs}Q^{-1}(s)b(s)~ds\\
&=e^{BT}z(t)+e^{B(t+T)}\int_{t}^{t+T}e^{-Bs}Q^{-1}(s)b(s)~ds\\
&=e^{BT}z(t)+e^{B(t_0+T)}\int_{t_0}^{t_0+T}e^{-Bs}Q^{-1}(s)b(s)~ds,
\end{align*}
applying the $T$-periodicity of $Q$ and $b$. The lemma follows.
\end{proof}

To derive a version of Claim~\ref{claim:main_heuristic} in finite dimensions, we need to allow the forcing $b(t)$ of \eqref{eq:floquetdrop} to vary from one period to the next. We can then use a technique reminiscent of that of \citet{durey_chaos} to characterize the mean dynamics of such a system:



\begin{theorem}\label{thm:finitedim}
Suppose $X(t)\in\CC^m$ is piecewise constant on each period $[k,k+T)$, for $k\in \{0,T,2T,...\}$, and undergoes an ergodic process with stationary probability density $\mu$ on $\CC^m$. Suppose that $y(t)\in\CC^{d}$ solves 
\begin{equation}\label{floquetdropx}
\dot{y}=L(t)y+b(t,X(t)),
\end{equation}
where $L(t)\in\CC^{d\times d}$ is $T$-periodic and continuous with monodromy matrix $\|e^{BT}\|<1$, and $b(t,x)$ is $T$-periodic and continuous in $t$ for each fixed value $x\in\CC^m$. Further suppose that $x\mapsto \int_0^T|b(t,x)|\,dt$ is uniformly bounded for $x\in\CC^m$. Then, the mean field $\langle y\rangle \doteq\lim_{N\rightarrow\infty}N^{-1}\sum_{n=1}^N y(nT)$ is given by
\[\langle y\rangle = \int_{\RR}\mu(X')y_B(X')dX',\]
where $y_B(X')\in\CC^n$ is the mean field when $X(t) = X'$ is held constant.\label{thm:main}
\end{theorem}

\begin{remark} In the walking droplet system, the term $y_B(X')\in\CC^d$ represents the wave field generated by a fixed bouncer at $X'\in\CC^m$; note that this term can be computed numerically for a given droplet model. The hypothesis that $\|e^{BT}\|<1$ corresponds to a subcritical forcing frequency $\Gamma<\Gamma_F$ in the droplet system, and the (relatively mild) integral hypothesis on $b$ simply guarantees that $R=R(t_0,x)$ is well-defined and uniformly bounded in $x$. 
\end{remark}
\begin{proof}
As before, write the fundamental matrix $Y$ of the homogeneous equation in Floquet form as $Y(t)=Q(t)e^{Bt}$. Without loss of generality, suppose $t\in\{0,T,2T,...\}$. Viewing the restriction $y|_{[t,t+T]}$ as a solution to~\eqref{floquetdropx} with periodic forcing, Lemma \ref{lem:floquetsol} implies
\begin{equation}
    y(t+T)=e^{BT}y(t)+R(X(t)),\label{eq:ziter}
\end{equation}
noting that the dependence of $b$ on $X$ confers the same to $R$. 
We can now take the mean value of both sides, noting (from our integrability assumptions on $b$) that $R$ is uniformly bounded in $X$; subsequently using the ergodic theorem~\citep{ergodictheorem} to simplify, we have
\[(1-e^{BT})\langle y\rangle=\int_\RR R(X')\mu(X')~dX'.\]
By hypothesis, $1-e^{BT}$ is invertible, and thus
\[\langle y\rangle=\int_\RR (1-e^{BT})^{-1}R(X')\mu(X')~dX'.\]
The stationary case corresponds to a Dirac measure $\mu(X)=\delta(x-X)$. Thus, $y_B(x,X)=(1-e^{BT})^{-1}R(X)$, proving the theorem. 
\end{proof}

The results of the present section can be applied to ordinary differential equations of the form \eqref{eq:floquetdrop} or \eqref{floquetdropx}. Moreover, we can also apply Theorem~\ref{thm:finitedim} without modification to linear, self-adjoint PDEs, by applying it separately to each eigenmode. The latter technique was applied by \citet{durey_chaos} to study walking droplets in uniform, unbounded domains\footnote{In fact, Durey, Milewski, and Bush studied a case where the operator decomposes into 2-D subspaces, but the application of Floquet theory is similar.}; \dave{their proof for general geometries uses more abstract techniques, like those of the following section}. 


\subsection{Statistical Correspondence in the PDE Case}\label{sec:infinitedim}
Theorem~\ref{thm:main} applies only to finite-dimensional ODEs (as well as self-adjoint PDEs, as discussed above); for instance, it can be used to prove a PDF-MWF relationship for a discretized droplet system, but not for the continuum limit. In order to proceed, we attempt to lift the argument of the preceding section to a slightly more abstract setting; one must interpret our PDE solution as an infinite-dimensional vector in a \emph{Banach space} $\cB$, and the wave operator $L$ as an unbounded operator on the same space. We introduce this formalism incrementally over the present section; despite these complications, the logic proceeds in much the same way as outlined above, and the results can be used with minimal changes.

First, a Banach space is defined as follows:

\medskip
\begin{definition}
    A \emph{Banach space} $\cB$ is a real or complex vector space, equipped with a norm $\|\cdot\|$ and with the property of `completeness'. That is, if $u_k\in\cB$ is a sequence for which $\|u_k-u_\ell\|\to 0$ as $k,\ell\to\infty$, then the sequence converges to a limit $u\in\cB$.\label{def:banach}
\end{definition}
\begin{example} Vector spaces of physical interest typically fall into this class:
\begin{enumerate}
    \item If $\Omega\subset\RR^d$ is a fixed domain, the space $L^\infty(\Omega;\RR^m)$ of bounded vector fields on $\Omega$ is a Banach space. For any bounded vector field $f:\Omega\to\RR^m$, we define the norm $\|f\|_{\infty} \doteq \sup_{x\in\Omega}|f(x)|$.
    \item In the same setting, the space $W^{\infty,k}(\Omega;\RR^m)$ of vector fields with bounded $k^{th}$ derivatives is a Banach space, with norm $\|f\|_{\infty,k} = \sum_{\ell\leq k}\|\nabla^{\otimes \ell} f\|_\infty$.
    \item Hilbert spaces---used widely in quantum mechanics---are special cases of Banach spaces. If $\langle\cdot,\cdot\rangle$ is the inner product on a Hilbert space $H$, the norm $\|\psi\|_H\doteq \sqrt{\langle\psi,\psi\rangle}$ furnishes $H$ with a Banach space structure.
\end{enumerate}
\end{example}
Notably, we can identify $L^\infty = W^{\infty,0}$, and the derivative $\nabla$ can be seen as a linear map from $W^{\infty,k}$ to $W^{\infty,k-1}$ for any $k$. To study a wave equation like ours~\eqref{eq:dureyeqs}, one is often interested in a slight generalization of the spaces $W^{\infty,k}$:
\begin{definition}
    If $1\leq p <\infty$, we define the Banach space $L^p(\Omega;\RR^m)$ of vector fields $f:\Omega\to\RR^m$ with 
    \[\|f\|_p \doteq \Big(\int_\Omega |f(x)|^p\,dx\Big)^{1/p}<\infty.\] 
    Likewise, we define the Banach space $W^{p,k}(\Omega;\RR^m)$ of vector fields with $k^\mathrm{th}$ derivatives in $L^p(\Omega;\RR^m)$, with norm $\|f\|_{p,k} = \sum_{\ell \leq k}\|\nabla^{\otimes \ell}f\|_p$.
\end{definition}

As the notation suggests, as $p\to\infty$, these norms converge to $\|\cdot\|_\infty$ and $\|\cdot\|_{\infty,k}$ for functions where the latter are defined. 

We can consider any linear, parabolic PDE (such as~\eqref{eq:dureyeqs}) as a differential equation on a Banach space $\cB$. We are interested in equations of the form
\begin{equation}\label{eq:banachODE}
    \dot{y}(t) = L(t)y + b(t,X(t)),
\end{equation}
where $L(t)$ and $b(t,X')$ are $T$-periodic for each $X'\in\RR^N$. In the walking droplet case, $y(t)\in\cB$ is our time-evolving wave field, $L(t)$ is our linear wave operator, $X(t)\in\RR^N$ is the configuration (position, velocity, phase, etc.) of the droplet on impact, and $b(t,X(t))\in\cB$ is the pressure applied by the droplet on $y$. In the droplet setting~\eqref{eq:dureyeqs}, we identify 
\[y = (\eta,\phi)\in W^{p,2}(\Omega;\RR^2),\]
for a fixed $1\leq p<\infty$, and thus restrict to the set of wave fields with well-defined second derivatives.

\begin{remark}
    An important note is, $L$ is \emph{not} generally well-defined as a map $\cB\to\cB$. For instance, if $\cB=W^{p,2}$ and $L$ involves two spatial derivatives, we can generically only say that $Ly\in L^p$ for any $y\in W^{p,2}$. This is a typical feature of differential equations; such an operator is said to be \emph{unbounded} on $\cB$. Otherwise, if $Ly\in\cB$ for all $y\in\cB$, the operator $L$ is said to be \emph{bounded}.
\end{remark}

Even with an unbounded operator, one often wants to keep working within the setting of $\cB$. To this end, we associate to any unbounded operator $L$ a subspace $\operatorname{dom}(L)\subset\cB$ such that, if $y\in\operatorname{dom}(L)$, we have $Ly\in\cB$; this subspace is known as the \emph{domain} of $L$. In the droplet setting, then, we fix
\[\cB = L^p(\Omega;\RR^2),\qquad \op{dom}(L) = W^{p,2}(\Omega;\RR^2)\subset\cB,\] 
so that $y\in\op{dom}(L)$. So long as $p<\infty$, the domain $\operatorname{dom}(L)$ has the property that it is `dense' in $\cB$. That is, for any $y\in\cB$, there is a sequence $y_k\in\operatorname{dom}(L)$ such that $y_k\to y$ as $k\to\infty$. In other words, any field in $L^p$ can be approximated to arbitrary accuracy by fields in $W^{p,2}$. This property is critical; it means that, even if $L$ cannot be applied everywhere in $\cB$, it can be applied \emph{nearly} everywhere.
    




Broadly speaking, the main difficulty in \dave{adapting Theorem~\ref{thm:finitedim} to} the infinite-dimensional case is that Floquet's theorem does not generally hold on a Banach space, preventing us from recovering an equivalent of Lemma~\ref{lem:floquetsol} directly. \dave{If our problem cannot be decomposed into ODEs as discussed in the preceding section}, one would like to proceed by trying to find a `Floquet form' appropriate \dave{for our PDE itself}. The first difficulty is, it is often non-trivial to show that a unique solution exists in the first place, even for (time-dependent) linear problems; we offer an easily-verified set of sufficient conditions in Proposition~\ref{prop:criteria} below, adapted from \citet{Pazy1963}. Even if a solution is known to exist, however, one must be careful in putting it into Floquet form. Recall that in the finite-dimensional case, the solution can be put in the form $y(t)\sim U(t)e^{Mt}y(0)$ for an invertible, periodic matrix $U$ (the `Floquet function') and a constant matrix $M$ (the `Floquet exponent'). To recover a periodic form of $U$, one must know that $e^{Mt}$ exists and is invertible for all $t$. This latter criterion can fail even for otherwise well-behaved PDEs~\citep{Kuchment1993}.

Fortunately, one does not need an exact Floquet decomposition to recover a PDF-MWF correspondence in the infinite-dimensional case. We proceed by leveraging the theory of \emph{evolution systems} \citep{Pazy1963} to characterize solutions to time-dependent parabolic PDEs:

\begin{definition}\label{def:evolutionsystem}
    Following \citet[Definition~5.3]{Pazy1963}, we define an \emph{evolution system} on a Banach space $\cB$ as a two-parameter family of bounded linear operators $U(t,s):\cB\to\cB$, for values $0\leq s\leq t\leq T$, satisfying
    \begin{enumerate}
        \item[$\mathbf{(E_1)}$] $U(s,s)=1$;
        \item[$\mathbf{(E_2)}$] $U(t,r)U(r,s)=U(t,s)$ for $s\leq r\leq t$; and
        \item[$\mathbf{(E_3)}$] For any $x\in\cB$, the map $(t,s)\mapsto U(t,s)x$ is continuous.
    \end{enumerate}
\end{definition}
The primary utility of evolution systems is in providing solutions to differential equations with time-dependent coefficients, akin to the fundamental (solution) matrix of finite-dimensional ODEs. Roughly, if $Y(t)$ is the fundamental matrix for a finite-dimensional ODE, the corresponding evolution system would take the form $U(t,s) = Y(t)Y(s)^{-1}$.

For a Banach space $\mathcal{B}$ and a family of operators $L(t):D\rightarrow\mathcal{B}$ on a uniform, dense domain $D\subset\mathcal{B}$, consider the initial value problem
\begin{equation}
    \dot y(t)=L(t)y, \qquad y(0)=y_0\in D.\label{eq:inftyfloquet}
\end{equation}
We say that~\eqref{eq:inftyfloquet} `admits an evolution system $U(t,s)$' if, for any $x\in D$, we have 
\begin{equation}\label{eq:evolutionsystem}
    \begin{aligned}
    \partial_t U(t,s)x &= L(t)U(t,s)x,\\
    \partial_s U(t,s)x &= -U(t,s)L(s)x,
    \end{aligned}
\end{equation}
for $0\leq s\leq t\leq T$. In this setting, we say that $y(t) = U(t,0)y_0$ is a `classical solution' to~\eqref{eq:inftyfloquet}. Following \citet[Theorem 6.1]{Pazy1963}, we propose the following criteria for the equation~\eqref{eq:inftyfloquet} to admit a unique evolution system. The following can be seen as an analogue to the classical Hille--Yosida theorem \citep{Reed1975} for time-dependent parabolic PDEs:
\begin{proposition}\label{prop:criteria}
    The equation~\eqref{eq:inftyfloquet} admits a unique evolution system $U(t,s)$ if the following conditions are met:
    \begin{enumerate}
        \item[$\mathbf{(P_1)}$] The domain $D$ is dense in $\mathcal{B}$ and independent of $t$.
        \item[$\mathbf{(P_2)}$]  There are constants $c,M>0$ such that for all $t\in[0,T]$ and for all $\lambda$ with $\Re~\lambda\geq c$, the resolvent $R(\lambda:L(t))\doteq (L(t)-\lambda)^{-1}$ exists and has operator norm\footnote{Recall that the \emph{operator norm} of a bounded operator $L:\cB\to\cB$ is the smallest $C>0$ such that, for all $y\in\cB$ with $\|y\|\leq 1$, we have $\|Ly\|\leq C$.} bounded as
        \[\norm{R(\lambda:L(t))}\leq\frac{M}{|\lambda-c|+1}.\]
        \item[$\mathbf{(P_3)}$]  The map $t\mapsto L(t)$ is H\"older continuous in operator norm\footnote{Specifically, this means that constants $C,\alpha>0$ exist such that, for any times $t_1$ and $t_2$, the operator $L(t_2)-L(t_1)$ has operator norm bounded by $\|L(t_2)-L(t_1)\|\leq C|t_2-t_1|^{\alpha}$.}. 

    \end{enumerate}
\end{proposition}
\begin{remark}
    The requirement $\mathbf{(P_3)}$ can be significantly weakened, although it is not necessary for the application to walking droplet dynamics. In general, the proposition still holds if $(\mathbf{P_3})$ is weakened to
    \[\normno{(L(t) - L(s))(L(\tau)-c)^{-1}}\leq C|t-s|^\alpha,\]
    for constants $C>0$ and $\alpha\in(0,1]$ and all $t,s,\tau\in[0,T]$.
\end{remark}

\begin{proof}
    Define $\tilde{L}(t)=L(t)-c$, which carries the same domain $D$ as $L(t)$. From $\mathbf{(P_2)}$, note that $\tilde{L}(t)$ has a resolvent $R(\zeta:\tilde{L}(t))$ defined for all $\zeta\geq 0$. So, for all $\zeta=\lambda-c\geq 0$, we have
    \[\mathbf{(P_2)'}:\quad\normno{R(\zeta:\tilde{L}(t))} = \norm{R(\lambda:L(t))}\leq\frac{M}{|\zeta|+1},\]
    for some constant $M>0$. Now, $\mathbf{(P_3)}$ implies that $\tilde{L}(t)$ is also H\"older continuous, so we have
    \begin{align*}
    \mathbf{(P_3)'}:\quad&\normno{(\tilde{L}(t) - \tilde{L}(s))\tilde{L}(\tau)^{-1}}\\
    &\qquad\leq \normno{\tilde{L}(t) - \tilde{L}(s)}\normno{\tilde{L}(\tau)^{-1}}\leq CM|t-s|^\alpha,
    \end{align*}
    for constants $C>0$ and $s,t,\tau\in[0,T]$. The conditions $\mathbf{(P_1)}$, $\mathbf{(P_2)'}$, and $\mathbf{(P_3)'}$ imply that $\tilde{L}(t)$ admits a unique evolution system $\tilde{U}(t,s)$ \citep[Theorem~6.1]{Pazy1963}.

    Now, let $U(t,s)=e^{c(t-s)}\tilde{U}(t,s)$; we now show that $U(t,s)$ is the unique evolution system that $L(t)$ admits. Note that $\mathbf{(E_1)}$ is satisfied since $U(s,s)=\tilde{U}(s,s)=1$, $\mathbf{(E_2)}$ is satisfied since $U(t,r)U(r,s)=e^{c(t-r)}e^{c(r-s)}\tilde{U}(t,r)\tilde{U}(r,s)=e^{c(t-s)}\tilde{U}(t,s)=U(t,s)$, and $\mathbf{(E_3)}$ is satisfied since strong continuity is preserved when multiplying by the continuous scalar function $e^{c(t-s)}$. Moreover,
    \begin{align*}
    \frac{\partial}{\partial t}U(t,s)&=e^{c(t-s)}\frac{\partial}{\partial t}\tilde{U}(t,s)+ce^{c(t-s)}\tilde{U}(t,s)\\&=(\tilde{L}(t)+c)e^{c(t-s)}\tilde{U}(t,s)=L(t)U(t,s),\\
    \frac{\partial}{\partial s}U(t,s)&=e^{c(t-s)}\frac{\partial}{\partial s}\tilde{U}(t,s)-ce^{c(t-s)}\tilde{U}(t,s)\\&=-e^{c(t-s)}\tilde{U}(t,s)(\tilde{L}(s)+c)=-U(t,s)L(s),
    \end{align*}
    so $L(t)$ indeed admits the evolution system $U(t,s)$. That $U(t,s)$ is unique follows from running the same logic backwards; if there were a distinct evolution system $U'(t,s)$ for $L(t)$, we could construct a distinct evolution system $e^{-c(t-s)}U'(t-s)\neq \tilde{U}(t-s)$ for $\tilde{L}(t)$.
\end{proof}

The conditions in Proposition~\ref{prop:criteria} cover a wide range of physically-relevant wave operators. As an example, we consider the operator found in the model of \citet{milewski_galeano-rios_nachbin_bush_2015}:

\begin{example}\label{ex:molacek}
    In the equation~\eqref{eq:dureyeqs} of Milewski \emph{et al.}, the wave operator is given as follows:\begin{equation}\label{eq:operator}
\dave{L(t)=\begin{pmatrix}
    2\nu\nabla^2 & -g(t) + (\sigma/\rho)\nabla^2\\
    \operatorname{DtN} &  2\nu\nabla^2
    \end{pmatrix}.}
\end{equation}
    If we consider the case $\cB=L^\infty(\RR^2;\RR^2)$ of bounded flow fields on $\RR^2$, the domain of $L(t)$ would be $D=W^{\infty,2}(\RR^2;\RR^2)$, independent of $t$. As discussed before, this subspace is dense in $\mathcal{B}$, so it satisfies $\mathbf{(P_1)}$. The spectral condition $\mathbf{(P_2)}$ follows from splitting $L(t)$ into a time-independent component $L_0=L|_{\cos 4\pi t = 0}$ and a time-dependent component $L_1(t)=L(t)-L_0$. The spectrum of the former---which corresponds to a non-vibrating fluid bath---is supported in the negative half-plane, so $L_0$ satisfies $\mathbf{(P_2)}$ with $c=0$. Applying the perturbation result~\citep[Section 4, Theorem 3.17]{Kato1995}, one can see that $L(t)$ satisfies $\mathbf{(P_2)}$ with a potentially larger $c$. The equation satisfies $\mathbf{(P_3)}$ because the only time-dependent term is uniformly continuous; hence, $L(t)$ itself is H\"older continuous.
\end{example}

Even though Floquet theory does not generally apply in this setting, as discussed above, one can construct a practical analogue of Floquet's theorem suitable for our purposes:

\begin{lemma}\label{lem:floquetbanach}
    Suppose the differential equation~\eqref{eq:inftyfloquet} admits a unique evolution system $U(t,s)$, and suppose the operators $L(t)$ are $T$-periodic. Then, for any $n\in\{0,1,2,...\}$ and $t_0\in[0,T]$, we have
    \[U(nT+t_0,0) = C(t_0)^nU(t_0,0) = U(t_0,0)C_0^n,\]
    for fixed, bounded operators $C(t_0)$ and $C_0$ on $\cB$, where $C_0$ is independent of $t_0$. Equivalently, $y(nT+t_0) = C(t_0)^ny(t_0) = U(t_0,0)C_0^ny(0)$ for any classical solution y.
\end{lemma}
\begin{proof}
    Consider the family of operators $U(t+T,s+T)$. Since $L(t)$ is $T$-periodic, $U(t+T,s+T)$ satisfies~\eqref{eq:evolutionsystem}, so $U(t+T,s+T)$ is an evolution system of~\eqref{eq:inftyfloquet}. However, our evolution system is unique by hypothesis, so we have $U(t,s)=U(t+T,s+T)$. Hence,
    \begin{align*}
        U((n+1)T+t_0,0)&=U((n+1)T+t_0,nT+t_0)U(nT+t_0,0)\\
        &=U(T+t_0,t_0)U(nT+t_0,0)\\
        &=C(t_0)U(nT+t_0,0),
    \end{align*}
    noting that $C(t_0)\doteq U(t_0+T,t_0)$ depends only on $t_0$. Similarly, we have
    \begin{align*}
        U((n+1)T+t_0,0)&=U((n+1)T+t_0,T)U(T,0)\\
        &=U(nT+t_0,0)U(T,0)\\
        &=U(nT+t_0,0)C_0,
    \end{align*}
    where $C_0\doteq U(T,0)$ is independent of $t_0$. The lemma follows inductively.
\end{proof}

As in the finite-dimensional case, we can leverage the decomposition provided by Lemma~\ref{lem:floquetbanach} to characterize solutions of the inhomogeneous equation:
\begin{lemma}\label{lem:inftyfloquetstep}
    Suppose $L(t)$ is $T$-periodic and satisfies the criteria of Proposition~\ref{prop:criteria}, and that $b(t)\in D$ is a $T$-periodic, H\"older continuous family of vectors. Then there is a unique solution to~\eqref{eq:inftyfloquet}, and it satisfies
    \begin{equation}
        y\left((n+1)T+t_0\right) = C(t_0)y(nT+t_0) + R(t_0)\label{eq:inftysol}
    \end{equation}
    for any $n\in\{0,1,2,...\}$ and $t_0\in[0,T]$; here, $C(t_0)$ is as given in Lemma~\ref{lem:floquetbanach} and $R(t_0)\in \cB$ is a constant vector.
\end{lemma}
\begin{proof}
    \dave{Following}~\citet[Theorem~4.1]{Pazy1963}, the solution to the initial value problem
    \[\dot y(t)=L(t)y+b(t),\qquad y(0)=y_0\in D\]
    exists, is unique, and is given by 
    \[y(t)=U(t,0)y_0+\int_{0}^t U(t,r)b(r)~dr.\]
    Setting $t=nT+t_0$, we have
    \begin{align*}
        y\left(t+T\right)&=U\left(t+T,0\right)y_0+\int_{0}^{t+T}U(t+T,r)b(r)~dr\\\ &=C(t_0)U(t,t_0)y(t_0)+C(t_0)\int_{0}^{t}U(t,r)b(r)~dr\\
        &\qquad+C(t_0)\int_{t}^{t+T}U(t,r)b(r)~dr\\
        &=C(t_0)y(t)+C(t_0)\int_{t_0}^{t_0+T}U(t_0,r)b(r)~dr,
    \end{align*}
    noting that $U(t+T,r)=C(t_0)U(t,r)$ and $U(t,s)=U(t+T,s+T)$. 
\end{proof}

Finally, we recover our primary result:
\begin{theorem}
    Suppose $X(t)\in\Omega\subset\CC^m$ is piecewise constant on each period $[k,k+T)$, for $k\in \{0,T,2T,...\}$, and undergoes an ergodic process with stationary probability density $\mu$. Suppose that $y(t)\in\mathcal{B}$ solves 
\begin{equation}\label{floquetdropx_infty}
\dot{y}=L(t)y+b(t,X(t)),
\end{equation}
where $L(t)$ is a $T$-periodic operator satisfying the conditions of Proposition~\ref{prop:criteria}, and $b(t,x)$ is $T$-periodic and H\"older continuous in $t$ for each fixed value $x\in\Omega$. Further suppose that the map $x\mapsto\int_0^T\|b(t,x)\|\,dt$ is uniformly bounded over $\Omega$, and that the map $C_0:y(0)\mapsto y(T)$ has norm $\|C_0\|<1$. Then, the mean field $\langle y\rangle \doteq\lim_{N\rightarrow\infty}N^{-1}\sum_{n=1}^N y(nT)$ is given by
\[\langle y\rangle = \int_{\Omega}\mu(X')y_0(X')dX',\]
where $y_0(X')\in\cB$ is the mean field when $X(t) \equiv X'$ is held constant.\label{thm:main_infty}
\end{theorem}
\begin{proof}
    Without loss of generality, suppose $t\in\{0,T,2T,...\}$. As in the finite-dimensional case, Lemma~\ref{lem:inftyfloquetstep} gives us
    \begin{equation}
        y(t+T)=C_0y(t)+R(X(t)). \label{eq:inftyfloquetx}
    \end{equation}
    We can now take the mean value of both sides, noting that $R(X')$ is uniformly bounded in $X'$ from our uniform integrability assumption on $b$. Using the ergodic theorem to simplify, we have
    \[(1-C_0)\langle y\rangle=\int_{\Omega}R(X')\mu(X')~dX'.\]
    Since $\norm{C_0}<1$, we know that $1-C_0$ is invertible with continuous inverse. Hence,
    \[\langle y\rangle=\int_{\Omega}(1-C_0)^{-1}R(X')\mu(X')~dX'.\]
    The stationary case $X(t)=X'$ corresponds to the Dirac measure $\mu(X')=\delta(x-X')$, or $y_0(x,X')=(1-C_0)^{-1}R(X')$. Thus, $y_0(X')=(1-C_0)^{-1}R(X')$, proving the theorem.
\end{proof}

\bibliography{Bibliography}

\end{document}